\documentclass[12pt]{amsart}

\usepackage{amsmath}
\usepackage{amssymb}
\usepackage{bm}
\usepackage{geometry}
\usepackage{amsthm}
\usepackage{amscd}
\usepackage{graphicx}
\usepackage{tikz-cd}
\usepackage{mathtools}
\usepackage[mathscr]{euscript}
\usepackage{setspace}
\usepackage{cjhebrew}

\newcommand{\dwedge}{\mathbin{\rotatebox[origin=c]{90}{\larger[-2]$\gg$}}}

\newcommand{\wedgedot}{\mathbin{\rotatebox[origin=c]{-90}{\larger[-2]$\lessdot$}}}

\usepackage{stmaryrd}

\usepackage{graphicx}
\usepackage{tikz}

\def\fg{\mathfrak g}
\def\fh{\mathfrak h}

\def\L8{L_\infty}

\def\cL{\mathcal L}
\def\cM{\mathcal M}
\def\cS{\mathcal S}

\def\NN{\mathbb N}
\def\RR{\mathbb R}

\def\ZZ{\mathbb Z}

\def\sO{\mathscr O}

\def\fJ(E){\mathfrak E}
\def\fJ{\mathfrak J}

\usepackage{hyperref}
\hypersetup{
    colorlinks=true,
    linkcolor=blue,
    filecolor=magenta,      
    urlcolor=cyan,
    filecolor=red,
    citecolor=blue,
}

\graphicspath{ {Desktop/} }
\newtheorem{thm}{Theorem}[section]

\newtheorem{prop}[thm]{Proposition}
\newtheorem{lem}[thm]{Lemma}
\newtheorem{cor}[thm]{Corollary}
\theoremstyle{definition}
\newtheorem{defn}[thm]{Definition}
\newtheorem{ex}[thm]{Example}
\newtheorem{rmk}[thm]{Remark}

\DeclareMathOperator{\Tr}{Tr}

\newcommand{\Sym}{\mathrm{Sym}}

\makeatletter
\newcommand*{\doublerightarrow}[2]{\mathrel{
  \settowidth{\@tempdima}{$\scriptstyle#1$}
  \settowidth{\@tempdimb}{$\scriptstyle#2$}
  \ifdim\@tempdimb>\@tempdima \@tempdima=\@tempdimb\fi
  \mathop{\vcenter{
    \offinterlineskip\ialign{\hbox to\dimexpr\@tempdima+1em{##}\cr
    \rightarrowfill\cr\noalign{\kern.5ex}
    \rightarrowfill\cr}}}\limits^{\!#1}_{\!#2}}}
\newcommand*{\triplerightarrow}[1]{\mathrel{
  \settowidth{\@tempdima}{$\scriptstyle#1$}
  \mathop{\vcenter{
    \offinterlineskip\ialign{\hbox to\dimexpr\@tempdima+1em{##}\cr
    \rightarrowfill\cr\noalign{\kern.5ex}
    \rightarrowfill\cr\noalign{\kern.5ex}
    \rightarrowfill\cr}}}\limits^{\!#1}}}
\makeatother

\title{Adjoint $L_\infty$-Actions and Conserved Charges in GR}

\author{Changsun Choi}
\address{Department of Physics, Montana State University, Bozeman 59717 USA}
\email{changsunchoi@montana.edu}

\author{Ryan E.~Grady}
\address{Department of Mathematical Sciences, Montana State University, Bozeman 59717 USA}
\email{ryan.grady1@montana.edu}

\begin{document}

\begin{abstract} In this work we compute the conserved currents and charges associated to the action of an infinitesimal isometry (Killing field) in Einstein--Cartan--Palatini gravity.  We offer a new approach to these quantities through the formalism of $L_\infty$-algebras, using the work of \'{C}iri\'{c}, Giotopoulos, Radovanovi\'{c}, and Szabo, as well as Costello and Gwilliam. We demonstrate our approach by computing the entropy of the Schwarzschild and Kerr black holes.  Along the way, we prove a purely algebraic result about the existence and utility of  a higher, indeed, fully homotopy coherent, version of the adjoint action of an $L_\infty$-algebra.
\end{abstract}

\maketitle

\tableofcontents

\section{Introduction} 

This paper is organized around computing the conserved currents and charges associated to the action of an infinitesimal isometry (Killing vector field) in Einstein gravity.  This problem is well studied and has appeared in textbooks for 50+ years.  Our main contribution  lies  more in the use of the formalism and classical Noether Theorem (Theorem 12.4.1) of Costello and Gwilliam \cite{CG2}. This both provides a novel example of the machinery and sets the stage for further applications of the Costello--Gwilliam formalism to (perturbative) gravity.

The input for Costello and Gwilliam's theorem is the action of a local $L_\infty$-algebra on a classical BV theory presented by another local cyclic $L_\infty$-algebra.  That gravity can be presented as such is a result of  \'{C}iri\'{c}, Giotopoulos, Radovanovi\'{c}, and Szabo \cite{CGRS}, coupled with a few minor observations of our own (see Section \ref{sect:CGRS}).  Einstein gravity can be presented in the Einstein--Cartan--Palatini (ECP) framework, where dynamic metrics are replaced by dynamic tetrads (coframes) and spin connections.  Building on previous work of Cattaneo and Schiavina, e.g., \cite{CS1}, \'{C}iri\'{c}--Giotopoulos--Radovanovi\'{c}--Szabo construct an $L_\infty$ presentation for the BV extension of ECP.  We show that this $L_\infty$-algebra acts on itself via a higher form of the adjoint action.  So too does any subalgebra, in particular the subalgebra of vector fields.  To this end, our first main result is the following.

\begin{thm}[Theorem \ref{thm:main2}]
Let $\cM_\text{ECP}$ be the local cyclic $L_\infty$-algebra presenting Einstein--Cartan--Palatini gravity on the four-manifold $M$. Let $\xi \in \mathrm{Vect}(M)$ be a Killing vector and set $\Lambda =0$.  Then, on shell, the current associated to $\xi$, 
\[
J[\xi] \colon \Omega^\ast_M \to \sO_\text{loc} (\cM)[-1],
\]
can be expressed locally by
\[
J[\xi](\beta)=\beta\wedge d\mathcal{Q}[\xi]\quad{\rm with}\quad  \mathcal{Q}[\xi]= \frac{1}{2}  \Tr \left (\iota_\xi \omega \wedge e \wedge e \right ).
\]\end{thm}

While the expression for the conserved current/charge is well known, our method, which utilizes equivariant action functionals derived from $L_\infty$-actions, has not been applied in this area before. As an explicit demonstration we compute the black hole entropy for Schwarzschild and Kerr spacetimes. A different approach that utilizes $L_\infty$-algebras in a different way to compute charges in gravity, using them to define homotopy moment maps, is in the work of Blohmann \cite{Bloh}.

Further, while $\mathrm{Vect}(M)$ is an ordinary Lie algebra, the $L_\infty$-formalism is essential to connect with the formalism of \cite{CG2} and provide a non-trivial example of their classical Noether Theorem.  
Moreover, our observations at the level of the full $L_\infty$-algebra provide a link between the work of Cattaneo--Schiavina (\cite{CS1}, \cite{CS2}), \'{C}iri\'{c}--Giotopoulos--Radovanovi\'{c}--Szabo (\cite{CGRS}), and the setting of Costello--Gwilliam \cite{CG1}, \cite{CG2}.

Along the way, we define and prove properties of a fully $L_\infty$ version of the adjoint action.  We call this the \emph{infinity adjoint action} and it is a strict generalization of the 
constructions of Mehta--Zambon \cite{MZ12} and Vitagliano \cite{Vit15}. See Section \ref{sect:compare} for an explicit comparison.  More specifically, we prove the following.

\begin{thm}[Theorem \ref{thm:main}] Let $(\cM, \{\ell_k\},\langle-,-\rangle)$ be a cyclic $L_\infty$-algebra.  There is an action of $\cM$ on itself, the \emph{infinity adjoint action}, such that
\begin{itemize}
\item[(a)] The action is compatible with the pairing (symplectic) on $\cM$;
\item[(b)] The action is Hamiltonian (inner); and
\item[(c)] The equivariant action is of the form $S^\text{tot}=S^\cL + S^\cM$, with $S^\cM$ encoding the $L_\infty$-structure on $\cM$.  In particular, $S^\text{tot}$ recovers the underlying classical action defined by $\cM$ when the background fields are set to zero.
\end{itemize}
\end{thm}

The present work is another step towards utilizing \emph{factorization algebras} in classical (and perturbative quantum) gravity; see also the recent work of Dul on linearized gravity \cite{Dul}. In a somewhat complementary direction, from an AQFT perspective, there is a good amount of related work by Rejzner, e.g., joint with Gwilliam \cite{GR}, \cite{GR2} and Benini and Schenkel (with collaborators) \cite{BS}, \cite{BCS}.

Since one of our goals is to provide an interesting example of the connection between several mathematical approaches to (classical) field theory and the Batalin--Vilkovisky (BV) formalism, and  since the foundational source material is spread across more than 1000 pages of literature, including the nearly 800-page two-volume set \cite{CG1}, \cite{CG2} (which, thankfully, is well written and readable),  we have chosen to present rather extended preliminary sections.  We hope that these two preliminary sections (and later, Sections \ref{sect:currents} and \ref{sect:CGRS}) might prove useful to others. Sections \ref{sect:infinity} and \ref{sect:charges} are genuinely new contributions (though likely not surprising to some subset of experts).

\subsection{Acknowledgements}
The authors thank Michele Schiavina and Filip Dul for  helpful correspondence.  They also thank Neil Cornish for many useful interactions and support as CC's co-advisor.
REG is supported by the Simons Foundation under Travel Support/Collaboration 9966728, he additionally thanks Damien Calaque for many discussions about this material during a visit to Montpellier which was supported by an OCCIMATH mini-grant.


\section{Preliminaries I: Local $L_\infty$-algebras}

Note that throughout we use cohomological grading and shift conventions, e.g., $V[-1]$ is a shift of the cochain complex $V$  ``one to the right." We adopt the Koszul sign rule, i.e., $uv = (-1)^{\lvert u \rvert \lvert v \rvert} vu$ for the graded (co)commutative case. The notation $V^\vee$ will denote the graded dual. We also use the symbol $\epsilon$ in two ways to indicate signs: $\epsilon (\sigma)$ will denote the sign of a permutation (or of a $(p,q)$-shuffle considered as a permutation) and $\epsilon_{ijkl}$ will denote the Levi--Civita symbol.  The Einstein summation convention is also employed, so we will sum over repeated indices of tensors.

Finally, recall that for a graded vector space (cochain complex, dg module, etc.), 
\[
\mathrm{Sym} (V) := \bigoplus_{n \ge0} \mathrm{Sym}^n (V)
\]
is a cocommutative coalgebra with coproduct
\[
\Delta(v_1 \cdots v_n) = \sum \epsilon(\sigma) v_{\sigma(1)} \cdots v_{\sigma(p)} \otimes v_{\sigma(p+1)} \cdots v_{\sigma(n)}
\]
where the sum is over all $(p,q)$-shuffles $\sigma$ with $p+q = n$.  The symmetric algebra also has the structure of a (graded) commutative algebra, as does its completed version.  On ($\RR$-linear) duals, the latter will be defined by
\[
\widehat{\mathrm{Sym}} (V^\vee) := \prod_{n \ge 0} \mathrm{Hom} \left ((V)^{\otimes n} , \RR \right)_{S_n},
\]
where the subscript denotes coinvariants with respect to the symmetric group and the right-hand side is equipped with the topology appropriate to the situation.

For sections of vector bundles over a manifold, we will also use  $\otimes$ to indicate the completed projective tensor product.


\subsection{Local $L_\infty$-algebras and Classical Field Theory}\label{sect:LCFT}

As opposed to working over a general ground ring, $R$, which is a commutative algebra over a characteristic zero field, we will be concrete and fix $R = \RR$.  Our application in later sections is over $\RR$ anyway.

\begin{defn}
An \emph{$L_\infty$-algebra consists} of a $\ZZ$-graded vector space, $\cL$, and a collection of $k$-ary multilinear maps $\{\ell_k\}_{k \ge 1}$,
\[
\ell_k \colon \cL^{\otimes k} \to \cL
\]
such that
\begin{itemize}
\item[(a)] For each $k \ge 1$, the map $\ell_k$ is graded anti-symmetric and of degree $2-k$, and
\item[(b)] The maps $\{\ell_k\}$ satisfy the \emph{generalized Jacobi identities}, i.e., for each $n \ge 1$
\[
\sum_{k=1}^n(-1)^{k(n-k)}\sum_{\sigma \in \mathrm{Sh}(k,n-k)} \chi(\sigma)\, \ell_{n-k+1} \left ( \ell_k \left (v_{\sigma (1)}, \dotsc , v_{\sigma (k)} \right ), v_{\sigma (k+1)} , \dotsc v_{\sigma (n)} \right ) =0.\]
\end{itemize}
\end{defn}

In the above, the sign is defined as \[\chi(\sigma)=\prod\left\{ -(-1)^{|v_i||v_j|} \,| \,i<j, \sigma(i)>\sigma(j) \right\}.\]

Any graded Lie algebra is an $L_\infty$-algebra as is any differential graded Lie algebra (dgla).  In the former case, only $\ell_2$ is nonzero, while in the latter, there can be nonzero $\ell_1$ and $\ell_2$.  In general, the multilinear maps $\{\ell_k\}$ are referred to as ``(higher) brackets."

Following \cite{CG2}, we will be interested in a particular type of sheaf of $L_\infty$-algebras parametrized by a spacetime manifold $M$.

\begin{defn}
Let $M$ be a manifold.  An $L_\infty$-algebra, $(\cL,\{\ell_k\})$, is \emph{local (on $M$)} if
\begin{itemize}
\item[(a)] $\cL$ is the space of sections of a graded vector bundle $L$ on $M$; and
\item[(b)] All brackets $\ell_k \colon \cL^{\otimes k} \to \cL$ are given by poly-differential operators.
\end{itemize}
\end{defn}

One of the most significant examples of a local $L_\infty$-algebra is the dgla of Lie algebra valued forms on an oriented $n$-manifold, $M$, $\cL = \Omega^\ast(M, \fg)$, where $\fg$ is an ordinary Lie algebra and the grading corresponds to form degree. The underlying vector bundle is $L = \Lambda^\ast(T^*M) \otimes \underline{\fg}$, where $\underline{\fg}$ is the trivial bundle with fiber $\fg$.  Note that if $M$ is a closed 3-manifold and $\fg$ is semi-simple, then we have a degree $-3$ pairing
\[
\langle-,-\rangle_L \colon L \otimes L \to \mathrm{Dens}_M [-3], \quad \alpha \otimes \beta \mapsto \Tr(\alpha \wedge \beta),
\]
where $\Tr$ is induced by the Killing form on $\fg$ and $\mathrm{Dens}_M$ is the density line bundle on $M$ (which in this case is trivialized via the orientation).  Moreover, this pairing is \emph{non-degenerate} and \emph{invariant}.  That is, the pairing induces a vector bundle isomorphism
\[
L \to L^\vee \otimes \mathrm{Dens}_M[-3]  \tag{\text{Non-Degeneracy}},
\]
and for each $n$ and all sections $\alpha_1, \dotsc \alpha_{n+1}$, the induced pairing on sections (which we simply denote by $\langle-,-\rangle$ with no subscript) is \emph{compatible} with the $L_\infty$ brackets, i.e., the map
\[
\cL^{\otimes n+1} \to \RR, \quad \alpha_1 \otimes \dotsb \otimes \alpha_n \otimes \alpha_{n+1} \mapsto \langle \ell_n (\alpha_1 , \dotsc , \alpha_n), \alpha_{n+1} \rangle \tag{\text{Invariance}}
\]
is graded anti-symmetric in all the $\alpha_i$.  

%

\begin{defn}
Let $\cL$ be a local $L_\infty$-algebra on $M$.  The algebra $\cL$ is \emph{cyclic of degree} $k$ if it is equipped with a degree $k$  non-degenerate symmetric pairing of vector bundles
\[
\langle -,-\rangle_L \colon L \otimes L \to \mathrm{Dens}_M [k]
\]
such that the induced pairing on compactly supported sections $\langle -, - \rangle \colon \cL_c \otimes \cL_c \to \RR$ is an invariant pairing. Here, symmetric means $\langle u,v  \rangle=(-1)^{(|u|+k)(|v|+k)}\langle v,u \rangle$.
\end{defn}

In what follows, all of our cyclic algebras will have pairings of degree $-3$.

Given a cyclic, local $L_\infty$-algebra, $(\cM, \{\ell_k\},\langle-,-\rangle)$, (on the manifold $M$, with finitely many nonzero brackets) we can use the shift $s:\cM\to\cM[1]$ to obtain a cyclic, local coalgebra $(\cM[1], \{\delta_k\},\langle-,-\rangle^{sh})$, where
\[
s\ell_k(v_1,...,v_k)=\delta_k\circ s^{\otimes k}(v_1,...,v_k)\quad{\rm and}\quad \langle\alpha,\beta\rangle^{sh}_L=(-1)^{|\alpha|}\langle s^{-1}\alpha, s^{-1}\beta\rangle_L.
\] The induced shifted pairing is defined similarly. Then, there is an associated \emph{Lagrangian density} on compactly supported sections. (We are denoting our $L_\infty$-algebra by $\cM$ at this point as later there will be two $L_\infty$-algebras and they play slightly different roles; this notation will be consistent with later sections.)   Indeed, for $\alpha \in \cM_c[1]$,
\[
\mathsf{L}(\alpha) = \sum_{k \ge 1} \frac{1}{(k+1)!} \langle \alpha, \delta_k (\alpha, \dotsc, \alpha) \rangle^{sh}_L \in \mathrm{Dens}_M.
\]
Correspondingly, there is an action functional $S \colon \cM_c[1] \to R$, given by
\[
S_\cM (\alpha) = \int_M \mathsf{L}(\alpha) = \sum_{k \ge 1} \frac{1}{(k+1)!} \langle \alpha, \delta_k (\alpha, \dotsc, \alpha) \rangle^{sh}.
\]

\begin{defn}
A local, cyclic $L_\infty$-algebra, $(\cM, \{\ell_k\},\langle-,-\rangle)$, (on the manifold $M$, with finitely many nonzero brackets) \emph{presents a classical field theory} with space of fields $\cM[1]$ and action $S_\cM \colon \cM_c[1] \to \RR$.
\end{defn}

In the case from before, $\cM := \Omega^\ast (M, \fg)$, with $M$ a closed oriented 3-manifold and $\fg$ semi-simple, the $L_\infty$-algebra $\cM$ is a presentation for a version of perturbative Chern--Simons theory.  Indeed, for $(\cM, \{\ell_k\}, \langle -,- \rangle)$ a local, cyclic $L_\infty$-algebra, the Euler--Lagrange equations for $S_\cM$ are the (generalized) Maurer--Cartan equations 
\[
\delta S_\cM(\alpha) = 0 \,\,\Leftrightarrow \,\sum_{k \ge 1} \frac{1}{k!} \delta_k (\alpha, \dotsc , \alpha) =0\,
\Leftrightarrow \,\sum_{k \ge 1}(-1)^{\frac{1}{2} k(k-1)} \frac{1}{k!} \ell_k  (\alpha, \dotsc , \alpha) =0
\] for degree zero element $\alpha\in \cM[1].$
In the above the sign comes from  $\delta_k(sv_1,...,sv_k)=(-1)^{(k-1)|v_1|+\cdots (k-k)|v_k|}s\ell_k(v_1,...,v_k)$ and $|\alpha|=|sv_i|=0$ for $1\leq i\leq k.$ Hence, we see that the intersection of the Euler--Lagrange locus with the classical original fields (those of degree one before the standard [1] shift, after which they are of degree 0) in $\cM = \Omega^\ast (M, \fg)$ exactly correspond to flat connections on the trivial $G$-bundle over $M$, where $\fg$ is the Lie algebra of $G$.

The structure of homotopy algebras in classical (and quantum) field theory, e.g., $A_\infty$ and $L_\infty$-algebras, is well documented, especially theories in the formalism of Batalin and Vilkovisky (BV).  In this article, we are interested in a particular presentation of a particular theory, so we won't discuss issues of existence, equivalence, etc.  Such issues (under varying hypotheses and conventions) are well discussed in Chapter 4 of \cite{CG2}, \cite{HohmZwiebach17}, \cite{JurcoMacrelli19}, and \cite{JurcoKim20}.

Note that in BV theory, the classical action should satisfy a \emph{master equation}: $\{S,S\}=0$.  Indeed, the action $S_\cM$ associated to a local, cyclic $L_\infty$-algebra, $(\cM, \{\ell_k\},\langle-,-\rangle)$, does satisfy the master equation in the algebra of local functionals.  In the formalism of \cite{CG2}, this follows from the discussion in Sections 3.5 and 4.4.  As one would expect, the invariant pairing on fields induces a Poisson bracket on functionals.  We will return to this, and a related point, in a later section when we discuss \emph{equivariant actions}. 

\begin{rmk}\label{rmk:shift}
We shift our $L_\infty$-algebra down by one, so that in the motivating example, the classical fields live in degree zero, the ghosts live in degree $-1$, etc.; this convention agrees with \cite{CG2}.  One need not perform the shift, indeed \cite{CGRS} is in the unshifted convention.  Hence, in Sections \ref{sect:CGRS} and \ref{sect:charges} we will work unshifted.  The action functional picks up additional signs in the unshifted setting: 
\[
S_\cM (\alpha) = \sum_{k \ge 1} (-1)^{\frac{1}{2} k(k-1)}\frac{1}{(k+1)!} \langle \alpha, \ell_k (\alpha, \dotsc, \alpha) \rangle.
\]
\end{rmk}

\begin{rmk}
As opposed to assuming only finitely many nonzero brackets, we could instead insist that our $L_\infty$-algebra is nilpotent. While this is not the case in our theory of interest, this condition is common (and useful) in field theory, especially $\sigma$-models, e.g., \cite{GG1}.
\end{rmk}

\begin{rmk}
In many cases, e.g., the definition of classical field theory in \cite{CG2}, one asks that $(\cL, \ell_1)$ is actually an \emph{elliptic complex}.  For the de Rham complex (and generally, complexes obtained from the de Rham complex), ellipticity holds for Riemannian manifolds, but generally fails for Lorentzian manifolds.  When the complex of fields is elliptic, one has a rich structure on the algebra of local observables, e.g., a $P_0$-factorization algebra.  As our interest is in constructing specific observables (and only at the classical level), not the whole algebra, the lack of ellipticity won't be an issue.
\end{rmk}

\subsection{Chains and Cochains of an $L_\infty$-algebra}

Given an $L_\infty$-algebra, $(\cL, \{\ell_k\})$, there is a well-defined  Lie algebra  homology complex, $C_\ast (\cL)$, which generalizes that of the Chevalley--Eilenberg homology complex of an ordinary (differential graded) Lie algebra.

\begin{defn}
Let $(\cL, \{\ell_k\})$ be an $L_\infty$-algebra.  The \emph{Lie algebra homology complex}, denoted $C_\ast (\cL)$, is the differential graded cocommutative coalgebra $\Sym(\cL[1])$ equipped with the degree one coderivation $\delta$ whose restriction to cogenerators are given by the brackets $\{\ell_k\}$, i.e., for each $k \ge 1$ the map
\[
\delta_k \colon \Sym^k (\cL[1]) \to \cL[1]
\]
is related to $\ell_k$ as earlier.
\end{defn}

By construction the data of the coderivation is the same as the data of the brackets, and that $\delta^2=0$ is equivalent to the generalized Jacobi relation.  Indeed, $L_\infty$-algebras are sometimes defined as graded vector spaces with such a square zero coderivation on the total symmetric algebra, see \cite{LodayVallette}.  As such, definitions for $L_\infty$-algebra adjacent constructions are often given in terms of the dg coalgebra $C_\ast (\cL)$.

\begin{defn}
Let $\cL_1$ and $\cL_2$ be $L_\infty$-algebras.  \emph{A map of $L_\infty$-algebras}, $F \colon \cL_1 \rightsquigarrow \cL_2$, is given by a map $F \colon C_\ast (\cL_1) \to C_\ast (\cL_2)$ of dg coalgebras.  A map $F \colon \cL_1 \rightsquigarrow \cL_2$ is \emph{strict} if it is determined by the linear map  $F_1 \colon \cL_1 [1]= \Sym^1(\cL_1[1]) \to \cL_2[1]$.
\end{defn}


Below, we will formulate \emph{action complexes} in terms of the algebras of cochains on local $L_\infty$-algebras, so we recall those definitions next.  Further details are provided in Sections 3.4 and 3.5 of \cite{CG2}.

\begin{defn}
Let $\cL$ be a local $L_\infty$-algebra over the manifold $M$.  Define the \emph{Lie algebra cohomology complex of $\cL$} to be the dg algebra
\[
C^\ast (\cL) := \prod_{n \ge 0} \mathrm{Hom} ((\cL[1])^{\otimes n} , \RR)_{S_n},
\]
where the right-hand side is the $S_n$-coinvariants of continuous linear maps from the completed projective $n$-fold tensor product of $\cL$ (shifted) with itself.
\end{defn}

Note that an $L_\infty$ map $F \colon \cL_1 \rightsquigarrow \cL_2$ determines a map of dg algebras $C^\ast (\cL_2) \to C^\ast (\cL_1)$ (and vice-versa under appropriate finiteness assumptions). Often times we simply write $C^\ast (\cL) = \widehat{\Sym}((\cL[1])^\vee)$, which is correct, if a bit ambiguous.  The definition above specifies in which category and with respect to which topology we are completing the symmetric algebra. The differential on $C^\ast (\cL)$ is the ``usual" Chevalley--Eilenberg differential which is dual to the coderivation given in the definition of $C_\ast (\cL)$.  Next, we will kill constant functions. 

\begin{defn}
Let $\cL$ be a local $L_\infty$-algebra over the manifold $M$.  The \emph{reduced Lie algebra cohomology complex of $\cL$} is given by
\[
C^\ast_{\text{red}} (\cL) := \ker (C^\ast (\cL) \to \RR),
\]
where the map on the right-hand side is the projection to the $n=0$ factor (also known as the augmentation map).
\end{defn}

We would also like a version of Lie algebra cohomology which corresponds to \emph{local functionals}, not just all continuous linear maps.  To accomplish this we will consider infinite jet bundles, so if $L$ is (graded) vector bundle, over $M$, then
\[
J(L)^\vee := \mathrm{Hom}_{C^\infty_M} (J(L), C^\infty_M) \quad \text{ and } \quad \sO_{\text{red}}(J(L)) := \prod_{n >0} \mathrm{Hom}_{C^\infty_M} (J(L)^{\otimes n} , C^\infty_M)_{S_n},
\]
are (pre)sheaves on $M$.  Actually, they have more structure, since the algebra of differential operators on $M$, $D_M$, naturally acts on both of these objects.

\begin{defn}
Let $L$ be a graded vector bundle over the manifold $M$.  The \emph{space of local functionals on $L$} is given by
\[
\sO_{\text{loc}} (L) := \mathrm{Dens}_M \otimes_{D_M} \sO_{\text{red}} (J(L)).
\]
\end{defn}

The space of local functionals, $\sO_{\text{loc}}(L)$, is a quotient of the perhaps more familiar space of Lagrangian densities $\mathrm{Dens}_M \otimes_{C^\infty_M} \sO_{\text{red}} (J(L))$ where two densities which differ by a total derivative are identified.

Now, if $\cL$ is a local $L_\infty$-algebra over $M$, then $J(L)$ also has the structure of an $L_\infty$-algebra (in the category of $D_M$ modules with symmetric monoidal structure given by tensoring over $C^\infty_M$).  Hence, the dg algebras $C^\ast (J(L))$ and $C^\ast_{\text{red}} (J(L))$ exist.

\begin{defn}
Let $\cL$ be a local $L_\infty$-algebra over the manifold $M$.  The \emph{local Lie algebra cohomology complex of $\cL$} is given by
\[
C^\ast_{\text{red,loc}} (\cL) := \mathrm{Dens}_M \otimes_{D_M} C^\ast_{\text{red}}(J(L)).
\]
\end{defn}

It follows from our earlier discussion that if $\cL_1 \rightsquigarrow \cL_2$ is a map of local $L_\infty$-algebras which is an inclusion at the level of vector bundles, then there is a map of cochain complexes $C^\ast_{\text{red,loc}} (\cL_2) \to C^\ast_{\text{red,loc}} (\cL_1)$. (The process of taking the tensor product $\mathrm{Dens}_M \otimes_{D_M} -$ can fail to preserve algebra structures, so these are only cochain complexes.)

\begin{prop}[Lemma 3.5.4 of \cite{CG2}]
If $M$ is oriented and $\cL$ is a local $L_\infty$-algebra on $M$, then there is a natural quasi-isomorphism
\[
C^\ast_{\text{red,loc}} (\cL) \cong \Omega^\ast (M, C^\ast_{\text{red}} (J(L)))[\dim M].
\]
\end{prop}

%


\section{Preliminaries II: Actions of Local $L_\infty$-algebras}

The material of this section is a (mild) repackaging of Part 3 of \cite{CG2}.  

There are two complementary motivations for the definitions and constructions of this section. The first is the observation that if $\mathfrak{m}$ is a module for a Lie algebra $\fg$, then there is an extension of Lie algebras
\[
\mathfrak{m} \to \fg \ltimes \mathfrak{m} \to \fg,
\]
where $\fg \ltimes \mathfrak{m}$ has underlying vector space $\fg \oplus \mathfrak{m}$ and Lie bracket
\[
[(X_1,m_1),(X_2,m_2)]_\ltimes :=([X_1,X_2]_\fg , X_1 \cdot m_2 - X_2 \cdot m_1).
\]
One can verify that this construction can actually be enhanced to an equivalence between $\fg$-modules and certain Lie algebra extensions of $\fg$. Further, this notion of action as extension makes sense for dglas and $L_\infty$-algebras, and, as we discuss at the end of the section, has appeared in the literature for some time.  The definition we recall below will enhance this existing definition by 1) extending the domain of definition to include local $L_\infty$-algebras, 2) incorporating cyclic structure, and 3)  the module $\mathfrak{m}$ to be equipped with its own bracket structure.

Secondarily, if $(X,\omega)$ is a symplectic manifold and $\fg$ is a Lie algebra, an action of $\fg$ on $X$ \emph{compatible} with $\omega$ is a map of Lie algebras $\alpha \colon \fg \to \mathrm{SympVect}(X,\omega)$, where $\mathrm{SympVect}(X,\omega)$  is the Lie algebra of vector fields preserving $\omega$, i.e., the Lie derivative of $\omega$ along such a vector field vanishes. The action defined by $\alpha$ is \emph{Hamiltonian} if it actually factors through $C^\infty_X$, and hence contained in the image of Hamiltonian vector fields. In this latter case, the map $\widetilde{\alpha} \colon \fg \to C^\infty_X$ is the dual to the moment map if the Lie algebra action came from a Hamiltonian group action.

In our setting of interest, the object being acted upon is not a symplectic manifold, but rather a cyclic local $L_\infty$-algebra $\cM$. To such an object there is a ``generalized" space, $B \cM$, whose algebra of functions is given by $C^\ast (\cM)$.  This perspective is used throughout \cite{CG2}, see also \cite{GG2} for a more concrete description and relationship to Lie algebroids.
Hence, we will phrase actions of Lie algebras on $\cM$ as maps to vector fields/functions on $B \cM$, so the latter is a map $\widetilde{\alpha} \colon \fg \to C^\ast (\cM)$.

Our discussion of extensions leads to the following definition.


\begin{defn}[Definition 12.2.1 of \cite{CG2}]
Let $(\cM, \{\ell_k\},\langle-,-\rangle)$ be a local cyclic $L_\infty$-algebra and $\cL$ a local $L_\infty$-algebra over the same  manifold $M$. An \emph{action of $\cL$ on $\cM$} is an $L_\infty$-algebra structure on $\cL \oplus \cM$ which fits into an exact sequence of sheaves of $L_\infty$-algebras
\[
0 \to \cM \to \cL \oplus \cM \to \cL \to 0.
\]
\end{defn}
Here, exactness means exactness of the underlying sequences of sheaves of graded vector spaces, and the maps are strict $L_\infty$-morphisms.

As the authors note, one consequence is that $\cM$ is both an $L_\infty$-subalgebra and an ideal for the $L_\infty$-structure on $\cL \oplus \cM$.  Hence, for $n \ge 1$, the  brackets on $\cL \oplus \cM$, $\{\ell^\ltimes_k\}$, determine maps (and are determined by maps), which we abusively continue to denote by
\[
\ell^\ltimes_{r+s} \colon \cL^{\otimes r} \otimes \cM^{\otimes s} \to \cM\quad r\ge 0, s\geq 1.
\]

We can further ask that the action be compatible with the pairing, or equivalently that it preserve the cyclic structure on $\cM$. Following our discussion in Section \ref{sect:LCFT} we give the following definition.

\begin{defn}\label{def:comp}
An action of $\cL$ on $(\cM, \{\ell_k\},\langle-,-\rangle)$ \emph{is compatible with the pairing} if for all $s \ge 1$ and all compactly supported sections $\{X_1, \dotsc, X_r\}$ of $\cL$ and $\{m_1 , \dotsc , m_s , m_{s+1}\}$ of $\cM$, the expression
\[
\langle \ell^\ltimes_{r+s} (X_1 , \dotsc , X_r , m_1 , \dotsc , m_s),m_{s+1} \rangle
\]
is graded anti-symmetric with respect to permutation of the variables $m_1,...,m_s,m_{s+1}$.
\end{defn}


\begin{ex}\label{ex:dgAdj}
Let $(\cM, \{\ell_k\},\langle-,-\rangle)$ be a cyclic local $L_\infty$-algebra. Let $\cM^\llcorner$ be the cyclic local $L_\infty$-algebra which has the same underlying graded vector space, differential and pairing, but for which all (higher) brackets are zero, i.e., $\cM^\llcorner$ is simply a cochain complex with a pairing.  Then $\cM$ acts on $\cM^\llcorner$ via the \emph{dg-adjoint} action 
with brackets
\[
\ell^\ltimes_{r+1} \colon \cM^{\otimes r} \otimes \cM^\llcorner \to \cM^\llcorner, \quad \ell^\ltimes_{r+1} (m_1 , \dotsc , m_r, m_1') = \ell_{r+1} (m_1 , \dotsc , m_r, m_1'),
\]
for $r \ge 1$ and one additional bracket when $r=0$
\[
\ell_1^\ltimes \colon \cM^\llcorner \to \cM^\llcorner , \quad \ell_1^\ltimes (m') = \ell_1 (m').
\]
(So brackets vanish if there is more than one element from $\cM^\llcorner$.)
This action is a compatible action.
\end{ex}


\subsection{Action Complexes}

As a warmup, let $(\fg, \langle -,-\rangle)$ be an Abelian Lie algebra (a vector space) with a pairing (we will ignore degree for a moment).  The Lie algebra cochains, $C^\ast (\fg)$, is equipped with a Lie bracket coming from the pairing. Indeed, the pairing induces an isomorphism $\Lambda^k \fg^\vee \cong \Lambda^{k-1} \fg^\vee \otimes \fg$. So given $f,g \in \Lambda^\bullet \fg^\vee$, we view them as functionals $f,g \colon \Lambda^{\bullet-1} \fg \to \fg$.  We then define 
\[
(f \circ g) = \sum_{\sigma\in{\rm Sh}(r+1,s)} {\rm sgn}(\sigma)f( g(X_{\sigma(0)}, \dotsc, X_{\sigma(r)}), X_{\sigma(r+1)}, \dotsc , X_{\sigma(r+s)}),
\]
where $g$ has arity $r+1$ and $f$ has arity $s+1$.  Then define $[f,g] = f \circ g -(-1)^{rs} g \circ f$.  This makes $\Lambda^\bullet \fg^\vee$ a dgla with zero differential. Finally, note that an element $\ell \colon \Lambda^2 \fg \to \fg$ such that $[\ell,\ell]=0$ is precisely a Lie bracket on $\fg$ which is compatible with the pairing $\langle -,-\rangle$. 

Under appropriate grading conventions, $\ell$ from above corresponds to a Maurer--Cartan element.  (The appropriate shift is $[2+\deg (\langle-,-\rangle)]$.  In what follows, we will restrict to the case that the pairing has degree $-3$, so there will be many shifts to the right by 1 (a $[-1]$ shift)).  The situation for other degree pairings is similar with slightly different degree shifts.

Taking care with grading, signs, and internal differentials, the preceding paragraph can be replicated in the setting of $L_\infty$-algebras.  In particular, if $(\cM, \{\ell_k\},\langle-,-\rangle)$ is a cyclic $L_\infty$-algebra, then $C^\ast(\cM)[-1]$ is a differential graded Lie algebra.  Further, if $\cM$ is local, then $C^\ast_{\text{loc}}(\cM)[-1]$ (and its reduced version) are also dglas. 


We have already seen that an action of a local $L_\infty$-algebra $\cL$ on $\cM$ consists, in part, of a $L_\infty$-structure on $\cL \oplus \cM$.  If we further require the action to be compatible with the pairing on $\cM$, then such an action defines a Maurer--Cartan element 
\[
\alpha \in C^\ast_{\text{red,loc}} (\cL \oplus \cM)[-1] \subset C^\ast (\cL) \otimes C^\ast_{\text{red,loc}} (\cM)[-1].
\]
(In the local setting the tensor product appearing is the completed projective tensor product.) The appearance of reduced cochains is explained by the fact that we don't want curved actions with constant components.

Similarly, if we expect to have a short exact sequence (extension) of local $L_\infty$-algebras as in the previous section, we need $\cM$ to sit inside of the $L_\infty$-algebra $\cL \oplus \cM$ as a subalgebra (and an ideal).  In particular, the $L_\infty$-structure restricts to the one already on $\cM$.  This is imposed by asking that our Maurer--Cartan element is in the kernel of the natural map $\rm res_{\cM} \colon C^\ast_{\text{red,loc}} (\cL \oplus \cM) [-1] \to C^\ast_{\text{red,loc}}(\cM)[-1]$.

\begin{defn}\label{def:Ham}
Let $(\cM, \{\ell_k\},\langle-,-\rangle)$ be a local cyclic (of degree $-3$) $L_\infty$-algebra and $\cL$ a local $L_\infty$-algebra over the same  manifold $M$.  The \emph{Hamiltonian action complex} of $\cL$ acting on $\cM$ is given by
\[
\mathrm{Ham}(\cL,\cM) := \ker \left (C^\ast_{\text{red,loc}} (\cL \oplus \cM) [-1] \xrightarrow{\; \rm res_{\cM} \;} C^\ast_{\text{red,loc}}(\cM)[-1] \right ).
\]
\end{defn}

Note that Costello and Gwilliam refer to this complex as \emph{the complex of ``inner" actions} and denote it $\mathrm{InnerAct}(\cL,\cM)$.  The proposition that follows is one justification for our notational change.  Additionally, there are silly examples of Lie algebra actions, e.g., (algebraic) vector fields acting on the Abelian Lie algebra $\cM = \RR[x,y] \oplus \RR[x^\vee,y^\vee][-3]$, which are Hamiltonian but not inner in a naive sense.

Note that $C^\ast_{\text{red,loc}} (\cL)[-1]$ is a subcomplex of $\mathrm{Ham}(\cL,\cM)$ which is contained in the center with respect to the Lie bracket.

\begin{defn}
Let $(\cM, \{\ell_k\},\langle-,-\rangle)$ be a local cyclic (of degree $-3$) $L_\infty$-algebra and $\cL$ a local $L_\infty$-algebra over the same  manifold $M$.  The \emph{symplectic action complex} of $\cL$ acting on $\cM$ is given by
\[
\mathrm{Act}(\cL,\cM) := \mathrm{Ham}(\cL,\cM)/ C^\ast_{\text{red,loc}} (\cL).
\]
\end{defn}

The discussion to this point (and Section 12.2 of \cite{CG2}) can be summarized as the following.

\begin{prop}[Lemma 12.2.5 of \cite{CG2}]\label{prop:Act}
Let $(\cM, \{\ell_k\},\langle-,-\rangle)$ be a local cyclic (of degree $-3$) $L_\infty$-algebra and $\cL$ a local $L_\infty$-algebra over the same  manifold $M$.  A compatible action of $\cL$ on $\cM$ is equivalent to a Maurer--Cartan element in $\mathrm{Act}(\cL,\cM)$.
\end{prop}

As $\mathrm{Act}(\cL,\cM)$ is a quotient of $\mathrm{Ham}(\cL,\cM)$, Maurer--Cartan elements of $\mathrm{Ham}(\cL,\cM)$ map to Maurer--Cartan elements of $\mathrm{Act}(\cL,\cM)$; as one would expect, since Hamiltonian actions are in particular symplectic actions.  Lifting a Maurer--Cartan element of $\mathrm{Act}(\cL, \cM)$ to one in $\mathrm{Ham}(\cL,\cM)$ is obstructed by a class in $H^1 (C^\ast_{\text{red,loc}} (\cL))$.

Let us (briefly) explain how the work of this section is connected to the second point of motivation articulated in the previous section.
The punchline is that, suppressing shifts, a Maurer--Cartan element $S^\cL \in \mathrm{Act}(\cL,\cM)$ determines a map $\cL \to \mathrm{SympVect}(B \cM) \cong C^\ast_{\text{red,loc}} (\cM)$, while a Maurer--Cartan element $\widetilde{S^\cL} \in \mathrm{Ham}(\cL,\cM)$ is a map $\cL \to \sO_{\text{loc}}(B \cM) \cong C^\ast_{\text{loc}} (\cM)$.

The key to the preceding paragraph is a form of Koszul duality (see Section 11.2 of \cite{CG2}).  Recall that a map of (ordinary) Lie algebras $F \colon \fg \to \fh$ is the same as a map of dg algebras $F^\ast \colon C^\ast (\fh) \to C^\ast(\fg)$ which in turn can be realized as (and realizes) a Maurer--Cartan element $S^F \in C^\ast(\fg)\otimes \fh$, where $C^\ast(\fg)\otimes \fh$ is the dgla which combines the dg algebra structure on $C^\ast(\fg)$ with the Lie bracket on $\fh$ entirely analogously to how we previously described the structure of the dgla of forms valued in a Lie algebra $\Omega^\ast(M,\fg)$.  Also, note in this setting there is no difference between Maurer--Cartan elements in $C^\ast(\fg)\otimes \fh$ and $C^\ast_{\text{red}} (\fg) \otimes \fh$, the latter is more useful moving forward.

This equivalence just described persists to the level of $L_\infty$-algebras. Suppose the base manifold $M$ is a point. Now, $\mathrm{SympVect}(B\cM) \cong C^\ast_{\text{red}}(\cM)[-1]$ is a dgla (as we discussed at the beginning of the subsection), so a map of $L_\infty$-algebras $\cL \rightsquigarrow C^\ast_{\text{red}}(\cM)[-1]$ is a Maurer--Cartan element in 
\[
C^\ast_\text{red} (\cL) \otimes C^\ast_\text{red} (\cM)[-1]  \cong (C^\ast_\text{red}(\cL) \otimes C^\ast (\cM)[-1])/C^\ast_\text{red}(\cL)[-1] \cong \mathrm{Act}(\cL,\cM).
\]
Next, at the level of cochain complexes (but not dglas), we have 
\[
\mathrm{Act}(\cL,\cM) \cong C^\ast_\text{red}(\cL \oplus \cM)[-1] /\left (C^\ast_\text{red}(\cL)[-1] \oplus C^\ast_\text{red}(\cM)[-1] \right ),
\]
so lifting the map $\cL \rightsquigarrow \mathrm{SympVect} (B \cM)$ to a map $\cL \rightsquigarrow \sO(B \cM)$ (still at the level of complexes) should just be a lift of the Maurer--Cartan element to a Maurer--Cartan element in $C^\ast_\text{red}(\cL \oplus \cM)[-1] / (C^\ast_\text{red}(\cL)[-1])$. Unfortunately, while morally correct, these identifications don't hold at the dgla level as $C^\ast_\text{red} (\cM)[-1]$ is not an ideal of $C^\ast_\text{red} (\cL \oplus \cM)[-1]$.  To fix this technicality, we construct $\mathrm{Ham}(\cL,\cM)$ first as a kernel (so subalgebra, not quotient) and then construct the map to $\mathrm{Act}(\cL,\cM)$. Finally, to move to the more general setting where $M$ is not necessarily a point, we just need to use local cochains.


\subsection{Equivariant Action Functionals}

Previously, in Section \ref{sect:LCFT}, we saw that local cyclic $L_\infty$-algebras present classical field theories.  We will now see that compatible actions of $L_\infty$-algebras on such objects lead to equivariant actions/theories.  This just boils down to an interpretation of the work in the previous subsection.

Let $\cL$ be a local $L_\infty$-algebra acting (compatibly) on a local cyclic (of degree $-3$) $L_\infty$-algebra $(\cM, \{\ell_k\},\langle-,-\rangle)$.  So via Proposition \ref{prop:Act}, the action is encoded by a Maurer--Cartan element $S^\cL \in \mathrm{Act}(\cL,\cM)$.  Now, let\footnote{One should be careful if $M$ is not compact and take compactly supported sections at this point, but this causes no further difficulty so we suppress this issue.}
\[
S \in \sO_\text{loc} (\cM[1]) \subset C^\ast_\text{red,loc}(\cM)[-1] \subset C^\ast_\text{red,loc} (\cL \oplus \cM)[-1],
\]
be the classical action on $\cM[1]$ or equivalently, the Maurer-Cartan element that encodes the local $L_\infty$-algebra structure of $\cM$.  Consider the element
\[
S^\text{tot} := S^\cL + S \in C^\ast_\text{red,loc}(\cL \oplus \cM)[-1]/(C^\ast_\text{red,loc} (\cL)[-1]).
\]
That $S^\text{tot}$ is a Maurer--Cartan element, with respect to $d_\cL + d_\cM$ and the bracket induced by the pairing on $\cM$, follows from $S$ and $S^\cL$ being Maurer--Cartan elements.  Indeed, since $d_\cM = \{S,-\}$, we have that
\[
d_\cL S^\text{tot} +\frac{1}{2} \{S^\text{tot},S^\text{tot}\} = d_\cL S^\cL + d_\cM S^\cL + \frac{1}{2} \{S^\cL, S^\cL\} +\frac{1}{2} \{S,S\},
\]
where the last summand is zero as $S$ satisfies the original CME (equivalently, actually defines an $L_\infty$-structure).

To summarize, if $\cL$ acts (compatibly) on $\cM$ with corresponding Maurer--Cartan element $S^\cL$, and $S \in \sO_\text{loc} (\cM[1])$ is our original action functional/$L_\infty$-structure, then $S^\text{tot} = S^\cL + S$ is an action functional on the fields $\cL[1] \oplus \cM[1]$ and satisfies an equivariant CME. Since Hamiltonian actions are in particular compatible actions, similarly they define equivariant action functionals as well.

\begin{ex}\label{ex:adjFun}
Let us return to the setting of Example \ref{ex:dgAdj}, and  let us again specialize to the cyclic dgla $\cM := \Omega^\ast (M, \fg)$, where $M$ is a closed oriented 3-manifold and $\fg$ is a (finite dimensional) Lie algebra equipped with invariant pairing, e.g., $\fg$ is semi-simple and we consider the Killing form.  Our classical action is that of perturbative Chern--Simons, so 
\[
S(A) = \int_M \frac{1}{2} \langle A, d A \rangle + \frac{1}{6} \langle A, [A,A] \rangle .
\]
Considering the dg-adjoint case, where our classical theory is just given by $\cM^\llcorner$, the cyclic $L_\infty$-algebra with only one nonzero bracket, the differential $\ell_1=d$.  Therefore, our action is
\[
S^\llcorner (A) = \int_M \frac{1}{2} \langle A, dA \rangle.
\]
As the mixed brackets of our $L_\infty$-structure on $\cM \oplus \cM^\llcorner$ are the same as before, our $S^\cL$ doesn't change. Hence, our total action is given as
\[
S^\text{tot}_\text{dgAdj} (B,A) = S^\cL(B,A) + S^\llcorner(A) = \int_M \frac{1}{2} \langle A, dA \rangle + \frac{1}{6} \langle A, [B,A] \rangle.
\]

\end{ex}

\subsection{Quick Comparison to Others}\label{sect:compare}

As mentioned previously, there has been a notion of action by an $L_\infty$-algebra on different types of objects for more than 30 years.  Most relevantly, by considering extensions similar to above, Mehta and Zambon \cite{MZ12} gave a definition of an $L_\infty$-algebra action which subsumed the original definition of $L_\infty$-modules of Lada and Markl \cite{LM95}. Their definitions worked well for (finite dimensional) dg-modules and their Proposition 6.1 introduced the dg-adjoint action of Example \ref{ex:dgAdj}.  Next, Vitagliano \cite{Vit15} gave a definition of representation of Lie--Rinehart algebras that generalized Mehta--Zambon and the notion of representation up to homotopy of a Lie algebroid.  While finite dimensionality was not imposed, the algebras were not cyclic, nor local. (Vitagliano's construction is compatible with the formalism of \cite{GG2} which relates Lie algebroids and local $L_\infty$-algebras.) The definition we give above from Costello--Gwilliam is for local cyclic $L_\infty$-algebras with no assumptions of finite dimensionality and strictly generalizes the previous definitions.

To our knowledge the full infinity adjoint action which we describe next has not  explicitly appeared in the literature.

\section{The Infinity Adjoint Action}\label{sect:infinity}

In Example \ref{ex:adjFun}, we saw the equivariant action functional associated to an $L_\infty$-algebra acting on itself via the dg-adjoint action.  This action was a functional on two copies of fields and we typically view the fields coming from the $L_\infty$-algebra which is doing the action as ``background fields."  When we charge our system with respect to a zero background field, we would like to recover our original action coming from the cyclic $L_\infty$-algebra which is being acted upon.
The dg-adjoint action fails to accomplish this wish, but there is a third action which does: the \emph{infinity adjoint action}.

Returning to our running example where $\cM := \Omega^\ast (M, \fg)$, for $M$  a closed oriented 3-manifold and $\fg$ a semi-simple Lie algebra, for the infinity adjoint action, we will have an equivariant action functional

\[
S^\text{tot}_{\infty \text{Adj}} (B,A) = S^\cL(B,A) + S^\cM(A) = \int_M  \frac{1}{6} \langle A, [B,A] \rangle + \frac{1}{2} \langle A, dA \rangle + \frac{1}{6} \langle A, [A,A] \rangle .
\]

Our aim of the present section is to prove the following result.

\begin{thm}\label{thm:main} Let $(\cM, \{\ell_k\},\langle-,-\rangle)$ be a local cyclic $L_\infty$-algebra of degree $-3$.  There is an action of $\cM$ on itself, the \emph{infinity adjoint action}, such that
\begin{itemize}
\item[(a)] The action is compatible with the pairing on $\cM$;
\item[(b)] The action is Hamiltonian; and
\item[(c)] The equivariant action is of the form $S^\text{tot}=S^\cL + S^\cM$, with $S^\cM$ encoding the $L_\infty$-structure on $\cM$.  In particular, $S^\text{tot}$ recovers the underlying classical action defined by $\cM$ when the background fields are set to zero.
\end{itemize}
\end{thm}

We immediately obtain the following corollary.  While the corollary seems a bit redundant, as BV theories already contain their gauge symmetries, we will see that it is useful for computing conserved currents and charges.

\begin{cor}
Let $(\cM, \{\ell_k\},\langle-,-\rangle)$ present a classical BV theory.  Then, $\cM$ acts on this BV theory by classical symmetries, as does any local $L_\infty$-subalgebra of $\cM$.
\end{cor}

\subsection{Definition of Infinity Adjoint Action}

We will start with the case of an ordinary $L_\infty$-algebra $(\fg, \{\ell_i\})$.  Let us fix some notation: for $p_i = (X_i, u_i) \in \fg \oplus \fg$ a homogenous element, and $N \subsetneq \NN$ a finite subset, $A \subsetneq N$ a proper subset, denote
\[
p_i^A = \begin{cases} X_i & i \in A\\u_i & i \notin A. \end{cases}
\]

\begin{defn}
Let $(\fg, \{\ell_k\})$ be an $L_\infty$-algebra.  The \emph{infinity adjoint action} of $\fg$ on itself is given by the brackets 
\begin{align*}
\ell_k^\ltimes (p_1, \dotsc, p_k) :&=\Big( \ell_k (X_1 , \dotsc , X_k ),  \ell_k (X_1+u_1 , \dotsc , X_k +u_k)-\ell_k (u_1 , \dotsc , u_k )\Big)
\\&=\left ( \ell_k (X_1 , \dotsc , X_k ) , \sum_{A \subsetneq \{1, \dotsc , k\}} \ell_k \left (p_1^A, \dotsc p_k^A \right ) \right),    
\end{align*}

for $k \ge 1$, which sits in the short exact sequence of $L_\infty$-algebras
\[
0 \to \fg \xrightarrow{i} \fg \oplus \fg \xrightarrow{\pi} \fg \to 0
\]
where  $i$ and $\pi$ are the natural inclusion (into the second summand) and projection (onto the first factor).
\end{defn}

\begin{lem}
The brackets $\{\ell_k^\ltimes\}$ on $\fg \oplus \fg$ define an $L_\infty$-algebra. 
\end{lem}

\begin{proof}
By construction, $i$ and $\pi$ are strict maps of $L_\infty$-algebras, each $\ell_k^\ltimes$ is multilinear, graded anti-symmetric and of degree $2-k$.  Hence, for the $L_\infty$-structure to be well defined, we need only check the generalized Jacobi identities. 
Let $n \ge 1$ and write $\ell_k(X,\sigma)=\ell_k (X_{\sigma(1)}, \dotsc ,
X_{\sigma(k)} )$ and similarly for $\ell_k(p^A,\sigma)$ and $\ell^{\ltimes}_k(p,\sigma)$. Then,
\begin{small}
\begin{align*}
&\sum_{k=1}^n (-1)^{k(n-k)}\sum_{\sigma \in \mathrm{Sh}(k,n-k)} \chi (\sigma) \ell^\ltimes_{n-k+1} \left ( \ell^\ltimes_k ( p, \sigma) , p_{\sigma (k+1)} , \dotsc , p_{\sigma(n)} \right )\\
&=\sum_{k=1}^n (-1)^{k(n-k)}\sum_{\sigma \in \mathrm{Sh}(k,n-k)} \chi (\sigma) \ell^\ltimes_{n-k+1} \left( \left( \ell_k (X,\sigma),  \sum_{A \subsetneq \{ \sigma(1) , \dotsc , \sigma (k)\}} \ell_k ( p^A,\sigma )\right),
p_{\sigma(k+1)}, \dotsc p_{\sigma (n)}  \right)\\
&= (X,U),  
\end{align*}
\end{small}
where
\[
X = \sum_{k=1}^n (-1)^{k(n-k)}\sum_{\sigma \in \mathrm{Sh}(k,n-k)} \chi (\sigma) \ell_{n-k+1} \left ( \ell_k (X,\sigma), X_{\sigma (k+1)} , \dotsc , X_{\sigma (n)} \right ),
\]
and
\begin{multline*}
U= \sum_{k=1}^n(-1)^{k(n-k)} \sum_{\sigma \in \mathrm{Sh}(k,n-k)} \chi (\sigma) \left [  \sum_{C \subsetneq \{ \sigma (k+1) , \dotsc, \sigma (n) \}} \ell_{n-k+1} \left ( \ell_k (X,\sigma), p^C_{\sigma (k+1)} , \dotsc , p^C_{\sigma (n)} \right ) \right.\\
+ \left. \sum_{B \subseteq \{ \sigma (k+1) , \dotsc , \sigma (n)\}} \ell_{n-k+1} \left ( \sum_{A \subsetneq \{ \sigma (1) , \dotsc , \sigma (k)\}} \ell_k \left ( p^A,\sigma) \right ), p_{\sigma (k+1)}^B , \dotsc , p_{\sigma (n)}^B \right ) \right ].
\end{multline*}
Now, $X=0$ as it is precisely a Jacobi identity for our original $L_\infty$-algebra $\fg$.  Similarly,
\[
U = \sum_{D \subsetneq \{1, \dotsc, n\}} \sum_{k=1}^n(-1)^{k(n-k)} \sum_{\sigma \in \mathrm{Sh}(k,n-k)} \chi (\sigma) \ell_{n-k+1} \left ( \ell_k (p^D_{\sigma (1)} , \dotsc , p^D_{\sigma (k)} ), p^D_{\sigma (k+1)}, \dotsc , p^D_{\sigma (n)} \right ).
\]
So, $U=0$ too, as it is a sum of Jacobi relations for our original $L_\infty$-algebra $\fg$.
\end{proof}

From the definition of the infinity adjoint action, it follows that we can restrict the infinity adjoint action to an $L_\infty$-subalgebra. The proof is just repeated use of the verification in the case of an ordinary Lie algebra.

\begin{prop}
Let $(\fg, \{\ell_k\})$ be an $L_\infty$-algebra and $\fh \subseteq \fg$ a subalgebra.  The infinity adjoint action restricts to an $L_\infty$-action of $\fh$ on $\fg$.
\end{prop}

The extension of the infinity adjoint action to local $L_\infty$-algebras is clear.

\begin{defn}
Let $(\cM, \{\ell_k\})$ be a local $L_\infty$-algebra.  The \emph{infinity adjoint action} of $\cM$ on itself is given by the short exact sequence of sheaves of $L_\infty$-algebras
\[
0 \to \cM \xrightarrow{i} \cM \oplus \cM \xrightarrow{\pi} \cM \to 0,
\]
where $i$ and $\pi$ are induced by the inclusion and projection maps of the underlying bundles, i.e., $i$ is inclusion into the second factor of the Whitney sum of bundles and $\pi$ is projection onto the first factor. The $L_\infty$-structure on $\cM \oplus \cM$ is given by the infinity adjoint structure previously defined for non-local $L_\infty$-algebras. Observe that $i$ and $\pi$ are strict $L_\infty$-algebra maps.
\end{defn}

\begin{rmk}
The computation in the above lemma can be viewed as arising from a skew product construction. Namely we start with the linear isomorphism 
\[\Phi:\mathcal{M}\oplus\mathcal{M}\to \mathcal{M}\oplus\mathcal{M}, \quad  \Phi(X,u)=(X,X+u).\] The $L_{\infty}$-algebra structure of $\mathcal{M}$ is encoded in the coderivation $\delta_{\mathcal{M}}$ on the coalgebra ${\rm Sym}(\mathcal{M}[1])$ such that $\delta^2_{\mathcal{M}}=0.$ Then, $\Phi$ induces the coalgebra isomorphism 
\[
{\rm Sym}(\Phi):{\rm Sym}((\mathcal{M}\oplus\mathcal{M})[1])\to{\rm Sym}((\mathcal{M}\oplus\mathcal{M})[1]).
\] On the other hand we have the direct product coderivation $\delta_{\mathcal{M}}\oplus \delta_{\mathcal{M}}$ on the coalgebra ${\rm Sym}((\mathcal{M}\oplus\mathcal{M})[1])$ such that $(\delta_{\mathcal{M}}\oplus \delta_{\mathcal{M}})^2=0$.
Now define a coderivation $\delta^{\ltimes}$ on the coalgebra ${\rm Sym}((\mathcal{M}\oplus\mathcal{M})[1])$ by 
\[
{\rm Sym}(\Phi)\circ\delta^{\ltimes}=(\delta_{\mathcal{M}}\oplus \delta_{\mathcal{M}})\circ{\rm Sym}(\Phi).
\] Clearly, $(\delta^{\ltimes})^2=0$ and this yields a semidirect product structure of $L_\infty$-algebras on  $\mathcal{M}\oplus\mathcal{M}$, thus giving an exact sequence of $L_\infty$-algebras as in the above definition. 
\end{rmk}

There is actually one thing to check for the infinity adjoint action of \emph{local} $L_\infty$-algebras to be well-defined: that the $L_\infty$-structure on $\cM \oplus \cM$ is \emph{local}.  That is, we need the brackets $\{\ell_k^\ltimes\}$ to be given by poly-differential operators.  This is a straightforward check, e.g., one could use Peetre's Theorem and properties of infinite jet bundles under Whitney sum. That is,  $\ell_k^\ltimes$  is built from the original polydifferential operators $\ell_k$ using bundle projections, inclusions, and finite sums. Since polydifferential operators are closed under these operations,  $\ell_k^\ltimes$
 is again polydifferential.

\begin{lem}
The infinity adjoint action of a local $L_\infty$-algebra on itself is local. Thus, the infinity adjoint action of local $L_\infty$-algebras is well-defined.
\end{lem}


\subsection{Properties of the Infinity Adjoint}

We now prove various properties of the infinity adjoint action, thereby proving Theorem \ref{thm:main}.  That the equivariant action functional induced by the infinity adjoint action restricts to the underlying classical action functional when the background fields are set to zero follows directly from the construction of the action  since $S^{\cL}$ contains at least one background field, so part (c) of the theorem is proved.  We still need to check parts (a) and (b) regarding compatibility (Definition \ref{def:comp}) and that the action is Hamiltonian (that it defines a Maurer--Cartan element in $\mathrm{Ham}(\cM,\cM)$ from Definition \ref{def:Ham}).  

While Hamiltonian implies symplectic, so (b) implies (a), we find it illustrative to prove compatibility first on its own. The proof is immediate, with the only subtlety arising from keeping track of the two copies of the $L_\infty$-algebra.

\begin{prop}
Let $(\cM, \{\ell_k\},\langle-,-\rangle)$ be a cyclic $L_\infty$-algebra. The infinity adjoint action of $\cM$ on itself is compatible with the pairing $\langle -, - \rangle$.
\end{prop}

\begin{proof}
This result is not dependent on the local structure, so it is enough to prove the claim for an ordinary $L_\infty$-algebra equipped with an invariant pairing. Let $(\fg, \langle -,- \rangle)$ be such an $L_\infty$-algebra and let $\{\ell^\ltimes_k\}$ denote the brackets for the infinity adjoint action of $\fg$ on itself. We need that for all $r \ge 0$, $s \ge 1$, and tuples $\{X_1 ,\dotsc , X_r\}$ from $\fg$ (the copy doing the acting) and $\{m_1, \dotsc , m_{s+1}\}$ from $\fg$ (the copy being acted upon)
\[
\langle \ell^\ltimes_{r+s} (X_1 , \dotsc , X_r , m_1 , \dotsc , m_s),m_{s+1} \rangle
\]
is graded anti-symmetric in all arguments. But this is clear as the bracket $\ell^\ltimes_{r+s}$ is application of the original bracket $\ell_{r+s}$ from $\fg$ after the ``concatenation" map $\fg^{\otimes r} \otimes \fg^{\otimes s} \to \fg^{\otimes (r+s)}$, and the $L_\infty$-algebra structure on $\fg$ itself is compatible with the pairing by hypothesis.
\end{proof}

\begin{prop}
Let $(\cM, \{\ell_k\},\langle-,-\rangle)$ be a cyclic $L_\infty$-algebra. The infinity adjoint action of $\cM$ on itself defines a Maurer--Cartan element in $\mathrm{Ham}(\cM,\cM)$, i.e., the infinity adjoint action is Hamiltonian.
\end{prop}

\begin{proof}
We will restrict to the case that the pairing has degree $-3$; the result follows for other degrees by carefully tracking various degree shifts. The proof will follow from some general considerations, see Lemma 12.2.3.4 of \cite{CG2}.  To begin, for any compatible action of $\cL$ on $\cM$, there is a short exact sequence of dg Lie algebras
\[
0 \to C^\ast_{\text{red,loc}} (\cL)[-1] \to \mathrm{Ham}(\cL,\cM) \to \mathrm{Act}(\cL,\cM) \to 0.
\]
Given a Maurer--Cartan element (so a compatible action) $S^\cL \in \mathrm{Act}(\cL, \cM)$ there is a naive lift $\widetilde{S^\cL} \in \mathrm{Ham} (\cL,\cM)$ coming from the fact that $C^\ast_\text{red} (\cM) \hookrightarrow C^\ast (\cM)$. The obstruction to $\widetilde{S^\cL}$ being a Maurer--Cartan element in $\mathrm{Ham}(\cL,\cM)$ can be viewed as an element
\[
\left. \left ( d_\cL \widetilde{S^\cL} + d_\cM \widetilde{S^\cL} + \frac{1}{2} \left \{ \widetilde{S^\cL} , \widetilde{S^\cL} \right \} \right ) \right |_{\sO_\text{loc}(\cL[1])} \in \sO_\text{loc} (\cL[1]).
\]
Such a functional always vanishes (for type reasons, see Lemma  12.2.3.4 of \cite{CG2}) for non-curved $L_\infty$-actions.  In the present setting, we are only considering $L_\infty$-algebras and actions without curving. So, in particular, the naive lift $\widetilde{S^\cL}$ of the infinity adjoint action is indeed a Maurer--Cartan element in $\mathrm{Ham}(\cM,\cM)$.
\end{proof}


\begin{rmk}
Curved $L_\infty$-algebras and actions show up often in field theory when one tries to model $\sigma$-models via $L_\infty$-spaces, see for instance \cite{GG1}.   Obstructions to Hamiltonian actions are also common when performing renormalization and/or quantization procedures, see Section 14.1.2 \cite{CG2} for a simple example.
\end{rmk}

\section{Conserved Currents and Charges from Local $L_\infty$-actions}\label{sect:currents}

We now pivot and recall the Noether Theorem for local $L_\infty$-actions as articulated in \cite{CG2}. See also Sections 4 and 5 of \cite{CG3} and Section 3 of \cite{Gw}. More specifically, we will outline the following result, which is really an amalgamation of the work done in previous sections and Section 12.5 of \cite{CG2}. 

\begin{prop}[Classical Noether Theorem of \cite{CG2}]\label{prop:noether}
Let $(\cM, \{\ell_k\},\langle-,-\rangle)$ be a local cyclic (of degree $-3$) $L_\infty$-algebra  over $M$ and $\fg$ an $L_\infty$-algebra (over a point). Furthermore, let $\fg$ act on $\cM$ via a Hamiltonian action.  Then, for each  degree-zero infinitesimal symmetry $X \in \fg$, there is an associated conserved current $J[X]$ of the classical field theory presented by $\cM$.
\end{prop}

Traditionally, if our classical theory lives over the manifold $M$ of dimension $d$, then the conserved current $J[X]$ is  a Lagrangian-valued $(d-1)$-form. Its integral over a (closed) codimension-one  submanifold $N \subset M$ is the \emph{charge} associated to the symmetry $X$. These quantities (and the statement of Noether's Theorem) are expressed via the \emph{variational bicomplex}. We will instead outline how to recover the same data in the setting of local $L_\infty$-algebras.

To begin, we need to resolve/extend the action of the (ordinary) $L_\infty$-algebra $\fg$ to that of a local $L_\infty$-algebra.  The following shows that there is essentially a unique (up to homotopy) way to do this.

\begin{lem}[Lemma 12.2.4.2 of \cite{CG2}]\label{lem:ext}
Let $(\cM, \{\ell_k\},\langle-,-\rangle)$ be a local cyclic (of degree $-3$) $L_\infty$-algebra  over $M$ and let $\fg$ be an $L_\infty$-algebra.  There is a canonical homotopy equivalence (of simplicial sets) describing 
\begin{itemize}
\item[(a)] Actions of $\fg$ on the classical field theory presented by $\cM$, and 
\item[(b)] Hamiltonian actions of the local $L_\infty$-algebra $\Omega^\ast_M \otimes \fg$ on $\cM$.
\end{itemize}
\end{lem}

So if $\fg$ acts via infinitesimal symmetries, then we have a Hamiltonian action of $\cL := \Omega^\ast_M \otimes \fg$ and correspondingly we obtain an equivariant action functional $S^\text{tot}$ which is a function of fields on $\cM$ and our background fields $\cL$.  We can use this equivariant action to extract a \emph{current map}
\[
J \colon \cL \to \sO_{\text{loc}}(\cM)[-1], \quad X \mapsto \frac{\delta S^\text{tot}}{\delta X}.
\]
(Passing to compactly supported sections of $\cL$, we obtain an observable $J[X] \in \mathrm{Obs}^\text{cl}$ via integration.)

Now let $X \in \fg$ and $\RR_X := \RR\langle X \rangle$ the Abelian Lie algebra of symmetries it generates. We then obtain a map 
\[
J[X] \colon \Omega^\ast_M \cong \Omega^\ast_M \otimes_\RR \RR_X \to \sO_\text{loc}(\cM)[-1].
\]
Let $N \subset M$ be a compact oriented codimension-one submanifold and $\eta_N \in \Omega^1_M$ its Poincar\'{e} dual.  The local functional $J[X] (\eta_N)$ can be identified with the  current evaluated on $N$ in the traditional sense, see Section 12.5 of \cite{CG2}, where it is shown that this current is also \emph{conserved}.

\begin{defn}
Let $X \in \fg$ be an infinitesimal symmetry of the classical field theory presented by the local cyclic $L_\infty$-algebra $(\cM, \{\ell_k\},\langle-,-\rangle)$. 
\begin{itemize}
\item[(a)] The \emph{current associated to $X$} is the map
\[
J[X] \colon \Omega^\ast_M \cong \Omega^\ast_M \otimes_\RR \RR_X \to \sO_\text{loc}(\cM)[-1].
\]
\item[(b)] If $N \subset M$ is a compact oriented submanifold, then the \emph{associated charge} is 
\[
\mathcal{Q}[X;N] := \int_M J[X](\eta_N),
\]
where $\eta_N$ is the Poincar\'{e} dual form to the submanifold $N$.
\end{itemize}
\end{defn}

\begin{rmk}
In \cite{CG2}, Proposition \ref{prop:noether} actually arises from several richer and more elegant mathematical constructions.  Most notably, we have chosen not to talk about the full structure present on the algebra of classical observables.  Had we recalled this structure in detail, we could have then stated Theorem 12.4.1 of \cite{CG2} which proves that the action of  $\cL$ actually determines a map of presheaves of  $L_\infty$-algebras $\cL \rightsquigarrow \mathrm{Obs}^\text{cl}$. Theorem 12.4.1.2 articulates how this map is compatible with (shifted) Poisson structures. Moreover, Chapter 13 of \cite{CG2} is an extension of these results to the quantum realm (given a quantization of the classical field theory).
\end{rmk}

\section{The \'{C}iri\'{c}--Giotopoulos--Radovanovi\'{c}--Szabo $L_\infty$-algebra}\label{sect:CGRS}

The Einstein--Cartan--Palatini (ECP) field theory is an approach to classical gravity.  In three and four dimensions, it is equivalent to the Einstein--Hilbert field theory, see \cite{CS3}, \cite{CS2}.  Cattaneo and Schiavina gave a BV extension of ECP theory in \cite{CS1} (actually they did more and provided a BV--BFV extension which allows for spacetimes with boundary). \'{C}iri\'{c}, Giotopoulos, Radovanovi\'{c}, and Szabo \cite{CGRS} then translated ECP into the language of $L_\infty$-algebras in dimensions three and four.  In this section, we recall the CGRS $L_\infty$-algebra $\cM_\text{ECP}$, focusing on the four-dimensional case with cosmological constant $\Lambda$ (Section 8 of \cite{CGRS}).  We also recall that it defines a cyclic local $L_\infty$-algebra as we have defined above; the salient details are already in \cite{CGRS}, we simply put them together carefully.

Throughout we will make repeated use of the standard isomorphism
\[
\Phi \colon \mathfrak{so}(1,3) \xrightarrow{\; \simeq \;} \Lambda^2 \RR^{1,3}, \quad \Phi (A)_{\mu \nu} = \eta(u_\mu, A u_\nu),
\]where $A \in \mathfrak{so}(1,3)$ is a matrix with  respect to the standard basis, $\{u_\nu\}$, of Minkowski space $\RR^{1,3}$ and $\eta$ is the Minkowski metric. Our four-dimensional spacetime, $M$,  will be a subset of $\RR^{4}$ with Lorentzian metric and we (occasionally) use global coordinates.

As a graded vector space the $L_\infty$-algebra is concentrated in four degrees 
\[
\cM_{\text{ECP}} := \cM^0 \oplus \cM^1 \oplus \cM^2 \oplus \cM^3, 
\]
with
\begin{align*}
&\cM^0 := \mathrm{Vect}(M) \oplus \Omega^0 (M, \mathfrak{so}(1,3)),\\
&\cM^1 := \Omega^1(M, \RR^{1,3}) \oplus \Omega^1 (M, \mathfrak{so}(1,3)),\\
&\cM^2 := \Omega^3(M, \Lambda^3 \RR^{1,3}) \oplus \Omega^3 (M, \Lambda^2 \RR^{1,3}),\\
&\cM^3:= \Omega^1(M, \mathrm{Dens}_M) \oplus \Omega^4(M,\Lambda^2 \RR^{1,3}).
\end{align*}
There are only three non-zero brackets: $\ell_1, \ell_2, \ell_3$. These brackets can be read off by decomposing the BV action of \cite{CS1} by polynomial degree in the fields. (See also \cite{HohmZwiebach17} for a general strategy.)

The differential, $\ell_1$, is nonzero on vectors of degree zero and two. For $(\xi, \rho) \in \cM^0$, $(E,\Omega) \in \cM^2$, 
\[
\ell_1 (\xi, \rho) = (0, d \rho) \quad \text{ and } \quad \ell_1 (E,\Omega) = (0,-d \Omega).
\]

Adding $(e,\omega) \in \cM^1$ and $(\mathcal{X},\mathcal{P}) \in \cM^3$ to the notation from the preceding paragraph, the binary bracket, $\ell_2$, is given by
\begin{align*}
&\ell_2 ((\xi_1,\rho_1),(\xi_2,\rho_2)) = ([\xi_1,\xi_2], -[\rho_1,\rho_2] + \xi_1 \cdot \rho_2 - \xi_2 \cdot \rho_1),\\
&\ell_2 ((\xi,\rho),(e,\omega))=(-\rho \cdot e+L_\xi e, -[\rho,\omega]+L_\xi \omega),\\
&\ell_2 ((\xi,\rho),(E,\Omega))= (- \rho \cdot E+L_\xi E, -[\rho,\Omega] + L_\xi \Omega),\\
&\ell_2 ((\xi,\rho),(\mathcal{X},\mathcal{P}))=(dx^\mu \otimes \Tr (\iota_\mu d\rho \dwedge \mathcal{P}) +L_\xi \mathcal{X}, -[\rho, \mathcal{P}]+L_\xi \mathcal{P}),\\
&\ell_2 ((e_1, \omega_1),(e_2,\omega_2))= - (e_1 \dwedge d \omega_2 + e_2 \dwedge d \omega_1 , e_1 \dwedge d e_2 + e_2 \dwedge d e_1), \text{ and}\\
&\ell_2 ((e,\omega),(E,\Omega)) = (dx^\mu \otimes \Tr(\iota_\mu de \dwedge E - \iota_\mu d\omega \dwedge \Omega - \iota_\mu e \dwedge dE + \iota_\mu \omega \dwedge d\Omega), \frac{3}{2} E \barwedge e + [\omega, \Omega]).
\end{align*}

These brackets contain a lot of notation, so let us recall some of it from Section 5 of \cite{CGRS}.  First, the ``double wedge" symbol, $\dwedge$, is the map
\[
\dwedge \colon \Omega^{k_1}(M,\Lambda^{l_1} \RR^{1,3}) \otimes \Omega^{k_2} (M, \Lambda^{l_2} \RR^{1,3}) \to \Omega^{k_1 + k_2} (M, \Lambda^{l_1 + l_2} \RR^{1,3}),
\]
which wedges the differential forms and the exterior algebra components. In \cite{CGRS}, this map is notated by $\curlywedge$, but we have chosen the double wedge symbol to emphasize its definition and to alleviate confusion for those of us with old eyeballs. (Many authors also simply use $\wedge$ for this operation.)  

The other wedge symbol, $\wedgedot$, is the exterior product of form components and the action of $\mathfrak{so}(1,3)$ on exterior algebra components (``the multivector representation"), e.g., below there is a term of the form $\omega \wedgedot e$, this ends up in $\Omega^2 (M, \RR^{1,3})$ as the form degrees add then the Lie algebra component of $\omega$ acts on the vector representation component of $e$. There is one additional term: $E \barwedge e$, where $E \in \Omega^3 (M, \Lambda^3 \RR^{1,3})$, $e \in \Omega^1 (M, \RR^{1,3})$, and where the result lives in $\Omega^4 (M, \Lambda^2 \RR^{1,3})$. This term is given by applying the Hodge dual (with respect to the Minkowski metric $\eta$), which is an isomorphism $\Lambda^3 \RR^{1,3} \cong \Lambda^{4-3} \RR^{1,3} \cong \RR^{1,3}$, to the exterior algebra component of $E$ and then applying ``$\dwedge$" with the form $e$. 

%


Throughout, ``$\cdot$" denotes the appropriate action, e.g., $\xi \cdot \rho$ is the vector field $\xi$ acting as a derivation of the vector-valued function $\rho$. As is standard, $L_\xi$ denotes Lie derivative. The operator $\iota_\mu$ is contraction with the coordinate vector field indexed by $\mu$.  The trace map, $\Tr$, is induced by the map $\Tr \colon \Lambda^4 \RR^{1,3} \xrightarrow{\; \simeq \;} \RR$ normalized by requiring that $\Tr(u_\mu \wedge u_\nu \wedge u_\alpha \wedge u_\beta) = \epsilon_{\mu \nu \alpha \beta}$.

Finally, there is one 3-ary bracket
\begin{align*}
 \ell_3 ((e_1, &\omega_1), (e_2, \omega_2), (e_3, \omega_3))\\
&= -\left (e_1 \dwedge [\omega_2 ,\omega_3] +e_2 \dwedge [\omega_1, \omega_3] + e_3 \dwedge [\omega_2 , \omega_1] + 3! \Lambda e_1 \dwedge e_2 \dwedge e_3, \right.\\
& \left. e_1 \dwedge (\omega_2 \wedgedot e_3) + {}_{(2 \leftrightarrow 3)}+e_2 \dwedge (\omega_1 \wedgedot e_3) + {}_{(1 \leftrightarrow 3)} + e_3 \dwedge (\omega_2 \wedgedot e_1) + {}_{(2 \leftrightarrow 1)} \right )
\end{align*}

That the brackets $\{\ell_1,\ell_2,\ell_3\}$ equip $\cM_{\text{ECP}}$ with the structure of an $L_\infty$-algebra is shown using two methods in the appendices of \cite{CGRS}.  First, the Jacobi relations are checked explicitly in Appendix A (in the three dimensional case, though the four-dimensional case is structurally similar).  Then, in Appendix B, the authors use a field theoretic (BRST--BV) argument to show that the $L_\infty$-structure is well-defined.

There is a pairing $\langle -, - \rangle \colon \cM_c \otimes \cM_c \to \RR$ of degree $-3$ given as follows
\[
\langle (e,\omega), (E,\Omega) \rangle := \int_M \Tr(e \dwedge E + \Omega \dwedge \omega), \quad (e,\omega) \in \cM^1, \;
(E,\Omega) \in \cM^2,\]
and
\[
\langle (\xi, \rho),(\mathcal{X}, \mathcal{P}) \rangle := \int_M \iota_\xi \mathcal{X} + \int_M \Tr(\rho \dwedge \mathcal{P}), \quad (\xi,\rho) \in \cM^0, \; (\mathcal{X}, \mathcal{P}) \in \cM^3.
\]
All other combinations of fields are zero for degree reasons.

\begin{prop}
The $L_\infty$-algebra $(\cM_{\text{ECP}}, \{\ell_1,\ell_2,\ell_3\}, \langle-,-\rangle)$ is a cyclic local $L_\infty$-algebra.
\end{prop}

\begin{proof}
As $\cM_{\rm ECP}$ is made from vector-valued tensor fields, it is clear that it arises as the sections of a graded vector bundle $E_\text{ECP} \to M$.  Moreover, the brackets $\{\ell_1, \ell_2, \ell_3\}$ are all built from poly-differential operators on $E_\text{ECP}$. The pairing on $\cM_\text{ECP}$ indeed arises from a non-degenerate pairing $E_\text{ECP} \otimes E_\text{ECP} \to \mathrm{Dens}_M$ induced by the fiberwise trace pairing on $\Lambda^\ast \RR^{1,3}$ and the canonical pairing between vector fields and 1-forms.  The non-trivial check is that $\langle -, - \rangle$ has the appropriate invariance/compatibility, but this is the content of Section 5.2 of \cite{CGRS}.
\end{proof}

Recall our discussion of shifts and signs from Remark \ref{rmk:shift}.
Thus, for the ECP $L_\infty$-algebra $\cM_{\rm ECP},$ the action looks like 
\[
S_\text{ECP} (\alpha) = \frac{1}{2} \langle \alpha, \ell_1 (\alpha) \rangle -\frac{1}{6} \langle \alpha, \ell_2 (\alpha, \alpha) \rangle - \frac{1}{24} \langle \alpha , \ell_3 (\alpha, \alpha, \alpha) \rangle.
\]

After restricting to classical fields $(e,\omega) \in \cM^1$, we obtain the equality
\begin{align*}
S_\text{ECP}(e,\omega)
&= \frac{1}{2} \langle (e,\omega), \ell_1((e,\omega)) \rangle
 -\frac{1}{6} \langle (e,\omega), \ell_2((e,\omega),(e,\omega)) \rangle \\
&\quad\quad -\frac{1}{24} \langle (e,\omega), \ell_3((e,\omega),(e,\omega),(e,\omega)) \rangle \\
&= \frac{1}{2} \langle (e,\omega), (0,0) \rangle
 -\frac{1}{6} \langle (e,\omega),  (-2 e \dwedge d\omega, -2e\dwedge de) \rangle\\&
 \quad\quad -\frac{1}{24} \langle (e,\omega), (-3e\dwedge [\omega,\omega]-6\Lambda e\dwedge e \dwedge e, -6 e\dwedge(\omega\wedgedot e)) \rangle \\
&= \int_M \operatorname{Tr} \left(\frac{1}{3}e\dwedge e\dwedge d\omega +\frac{1}{3}e\dwedge  de \dwedge \omega \right)+\int_M \operatorname{Tr} \left(\frac{1}{8}e\dwedge e\dwedge [\omega,\omega] +\frac{1}{4} e\dwedge (\omega\wedgedot de)\dwedge\omega \right)\\&\quad\quad +\int_M \operatorname{Tr} \left(\frac{\Lambda}{4}e\dwedge e\dwedge e\dwedge e\right)\\
&= \int_M \operatorname{Tr} \left(\frac{1}{3}e\dwedge e\dwedge d\omega -\frac{1}{3}\omega\dwedge e\dwedge  de \right)+\int_M \operatorname{Tr} \left(\frac{1}{8}e\dwedge e\dwedge [\omega,\omega] -\frac{1}{4}\omega\dwedge e\dwedge (\omega\wedgedot de) \right)\\&\quad\quad +\int_M \operatorname{Tr} \left(\frac{\Lambda}{4}e\dwedge e\dwedge e\dwedge e\right).
\end{align*}

Here, in this formal computation of the local functional, we work either with compactly supported fields/variations or, equivalently, in the quotient of Lagrangian densities by total derivatives. Hence, the boundary integral is zero and integration by parts gives
\[
0=\int_M d\,\operatorname{Tr} (e\dwedge e\dwedge \omega)=\int_M \operatorname{Tr}(2\,\omega\dwedge e\dwedge  de)+\int_M \operatorname{Tr} ( e\dwedge e\dwedge d\omega) ,
\] hence 
\[
\int_M \operatorname{Tr} \left(\frac{1}{3}e\dwedge e\dwedge d\omega -\frac{1}{3}\omega\dwedge e\dwedge  de \right)=\int_M \operatorname{Tr} \left(\frac{1}{2}e\dwedge e\dwedge d\omega \right).
\]

 Since ${\rm SO}(1,3)$ preserves $ u_0\wedge u_1\wedge u_2\wedge u_3$, $\mathfrak{so}(1,3)$ acts on $\Lambda^4\mathbb{R}^{1,3}$ trivially.  Thus, we get 
\[
e_1\dwedge e_2 \dwedge [\omega,\omega']=e_1\dwedge(\omega\wedgedot e_2)\dwedge\omega'+e_2\dwedge(\omega\wedgedot e_1)\dwedge\omega',
\] hence $e\dwedge e \dwedge [\omega,\omega]=2e\dwedge(\omega\wedgedot e)\dwedge\omega=-2\omega\dwedge e\dwedge(\omega\wedgedot e)$ and 
\[
\frac{1}{8}e\dwedge e\dwedge [\omega,\omega] -\frac{1}{4}\omega\dwedge e\dwedge (\omega\wedgedot de)=\frac{1}{4}e\dwedge e\dwedge [\omega,\omega].
\]
And, with $F_\omega=d\omega+\frac{1}{2}[\omega,\omega],$ our action becomes 
\begin{align*}
    S_\text{ECP}(e,\omega)&=\int_M \operatorname{Tr} \left(\frac{1}{2}e\dwedge e\dwedge d\omega 
    +\frac{1}{4}e\dwedge e\dwedge [\omega,\omega]+\frac{\Lambda}{4}e\dwedge e\dwedge e\dwedge e\right)\\&= \int_M \operatorname{Tr} \Bigl(
   \frac{1}{2} \, e \dwedge e \dwedge F_\omega
 + \frac{\Lambda}{4} \, e \dwedge e \dwedge e \dwedge e
\Bigr).
\end{align*}
That is, $S_\text{ECP}$ recovers the classical Einstein--Cartan--Palatini functional.

\section{Conserved Charges from the CGRS $L_\infty$-algebra: Black Hole Entropy}\label{sect:charges}

We are continuing the conventions of the previous section.  In particular, M is a four-dimensional spacetime, and the tetrads take values in the internal Minkowski space $\mathbb{R}^{1,3}$,  and $(\cM_{\text{ECP}}, \{\ell_1,\ell_2,\ell_3\}, \langle-,-\rangle)$ is the cyclic local $L_\infty$-algebra of \cite{CGRS}. The only exterior-type product used below is the one we denoted ``$\dwedge$" above, so as there is no possibility of confusion, we will follow the common convention and simply use ``$\wedge$" for this operation. Notice that $\omega\wedge\omega$ in $F_\omega$ has the usual meaning $\frac{1}{2}[\omega,\omega].$

By Theorem \ref{thm:main}, we know that $\cM_\text{ECP}$ acts on itself by the infinity adjoint action.  Moreover, this action restricts to any subalgebra.  We will consider the subalgebra $\mathrm{Vect (M)} \subset \cM_\text{ECP}$. Following Section \ref{sect:currents}, we will compute the conserved current and charge of the infinitesimal symmetry encoded by a vector field.


\begin{thm}\label{thm:main2}
Let $\xi \in \mathrm{Vect}(M)$ be a Killing vector and set $\Lambda =0$.  Then, on shell, the current associated to $\xi$, 
\[
J[\xi] \colon \Omega^\ast_M \to \sO_\text{loc} (\cM)[-1],
\]
can be expressed locally,  where $\xi\neq 0$, by
\[
J[\xi](\beta)=\beta\wedge d\mathcal{Q}[\xi]\quad{\rm with}\quad \beta\in\Omega^1(M)\quad {\rm and}\quad  \mathcal{Q}[\xi]= \frac{1}{2}  \Tr \left (\iota_\xi \omega \wedge e \wedge e \right )
\] for some coframe $e$ and torsion-free spin connection $\omega.$ 
\end{thm}

\begin{cor}[{\bf Noether  surface charge}]\label{cor:main}
Let $N \subset M$ be a compact oriented submanifold with boundary $\Sigma := \partial N$.  Let $\xi \in \mathrm{Vect}(M)$ be a Killing vector and set $\Lambda =0$. Assume that $\xi$ does not vanish on $N.$ Then, on shell, the charge associated to $\xi$ is
\[
\mathcal{Q}[\xi, N] = \frac{1}{2} \int_\Sigma \Tr \left (\iota_\xi \omega \wedge e \wedge e \right ).
\]
\end{cor}

 In the black-hole applications below, the corollary is applied to $N$  in the exterior region where $\xi\neq 0$ with $\partial N$ consisting of surfaces $\Sigma_r$ and $\Sigma_R$, $r<R.$ For the purpose of Wald's entropy we are only interested in the boundary component $\Sigma_r$ near 
the bifurcation-surface. After proving the theorem, we will compute the black hole entropy for the Schwarzschild black hole using the corollary.  After applying the theorem/corollary, this reduces to the standard textbook computation.  For us, it is not the result that is interesting, but rather that it can be done starting from the $L_\infty$-algebra formulation and applying the classical Noether Theorem of \cite{CG2}.

\subsection{Proving the Theorem}

To begin, we extend the action of $\mathrm{Vect}(M)$ to an action of $\Omega^\ast_M \otimes \mathrm{Vect} (M)$.  Lemma \ref{lem:ext} asserts that any extension which satisfies the master equation is unique, so we will use the ``minimal coupling." Hence, for $A_\xi \in \Omega^\ast_M \otimes \mathrm{Vect} (M)$, replacing $d$ in $S_{ECP}(e,\omega)$ by $d+A_{\xi}$, we have a total action
\[
S^\text{tot} (A_\xi, e, \omega) = S_{ECP} (e,\omega) + \frac{1}{6} \int_M \Tr \left ( 2e^2 \wedge A_\xi \cdot \omega - 2 \omega \wedge e \wedge A_\xi \cdot e \right ),
\]
where the action of $A_\xi$ is given by wedging the form component and acting via the Lie derivative by the vector field $\xi$ (as dictated by the bracket $\ell_2$).  

Hence, the 3-form component of the current $J[\xi]$, which we will denote by $J_3 [\xi]$, is given by
\[
J_3 [\xi] = {\rm Tr}\left[\frac{1}{3} e^2 \wedge L_\xi \omega - \frac{1}{3} \omega \wedge e \wedge L_\xi e\right].
\]

The theorem is proved by the following two lemmas.
\begin{lem}\label{lem:73} 
Let $\xi\in {\rm Vect}(M)$  and $\Lambda=0$. Then, on shell we have
\[
J_3[\xi]=d \mathcal{Q}[\xi]+\mathcal{A}[\xi],
\] 
where 
\[
\mathcal{A}[\xi]={\rm Tr}\left[-\frac{1}{6} e^2 \wedge L_\xi \omega - \frac{1}{3} \omega \wedge e \wedge L_\xi e\right].
\]
\end{lem}

\begin{lem}\label{lem:74}
Let $\xi$ be a Killing vector. Then,  locally, where $\xi\neq 0$, we have a coframe $e$ and a torsion-free spin connection $\omega$ such that $L_\xi e=0$ and $L_\xi\omega=0.$ 
\end{lem}

Indeed, locally, where $\xi\neq 0,$ Lemmas 7.3 and 7.4 give a coframe $e$ and a torsion-free spin connection $\omega$ such that   $J_3[\xi]=d \mathcal{Q}[\xi]$ with 
$\mathcal{Q}[\xi] = \frac{1}{2}  \Tr \left (\iota_\xi \omega \wedge e \wedge e \right ),$ proving the theorem. Thus, the corollary follows because, in view of Stokes's theorem, the charge associated to $\xi$ is
\[
\mathcal{Q}[\xi, N] = \int_N J_3[\xi]=\int_N d\mathcal{Q}[\xi]=\int_{\partial N} \mathcal{Q}[\xi]=\frac{1}{2} \int_\Sigma \Tr \left (\iota_\xi \omega \wedge e \wedge e \right ).
\]

\begin{proof}[Proof of Lemma \ref{lem:73}]
 
 Observe  that, if we consider only the Lie algebra part of $\omega$ in $\mathfrak{so}(1,3),$ then we have $e^{t\omega}\in{\rm SO}(1,3)$ for all $t\in\mathbb{R}$. Thus, for our orthonormal basis $u_0, u_1, u_2, u_3$ of the Minkowski space $\mathbb{R}^{1,3}$ we have 
 \begin{align*}
     e^{t\omega}u_a\wedge e^{t\omega}u_b\wedge e^{t\omega}u_c\wedge e^{t\omega}u_d&={(e^{t\omega})^p}_a{(e^{t\omega})^q}_b{(e^{t\omega})^r}_c{(e^{t\omega})^s}_d\,u_p\wedge u_q\wedge u_r\wedge u_s\\&={(e^{t\omega})^p}_a{(e^{t\omega})^q}_b{(e^{t\omega})^r}_c{(e^{t\omega})^s}_d\,\epsilon_{pqrs}\,u_0\wedge u_1\wedge u_2\wedge u_3\\&={\rm det}(e^{t\omega})\,\epsilon_{abcd}\,u_0\wedge u_1\wedge u_2\wedge u_3\\&=u_a\wedge u_b\wedge u_c\wedge u_d.
\end{align*} Upon differentiating at $t=0$ we get 
\[
(\omega u_a)\wedge u_b\wedge u_c\wedge u_d+u_a\wedge (\omega u_b)\wedge u_c\wedge u_d+u_a\wedge u_b\wedge (\omega u_c)\wedge u_d+u_a\wedge u_b\wedge u_c\wedge (\omega u_d)=0.
\] Now taking the trace we get
\[
0=-{\omega^p}_a\,\epsilon_{pbcd}- {\omega^p}_b\,\epsilon_{apcd}-{\omega^p}_c\,\epsilon_{abpd}-{\omega^p}_d\,\epsilon_{abcp}=d_\omega \epsilon_{abcd}.
\]

Thus, observing that $\mathcal{Q}[\xi]$ is a Lorentz scalar, we have 
\begin{align*}
2\, d\mathcal{Q}[\xi]&=d\,(\epsilon_{abcd} \,\iota_\xi\omega^{ab}\wedge e^c\wedge e^d) =d_\omega\,(\epsilon_{abcd} \,\iota_\xi\omega^{ab}\wedge e^c\wedge e^d) \\&=
\epsilon_{abcd} \,d_\omega(\iota_\xi\omega^{ab})\wedge e^c\wedge e^d+\iota_\xi\omega^{ab}(\epsilon_{abcd} \, d_\omega e ^c\wedge e^d) +\iota_\xi\omega^{ab}(\epsilon_{abcd} \, e^c\wedge d_\omega e^d)
\\&=\epsilon_{abcd} \,d_\omega(\iota_\xi\omega^{ab}) e^c\wedge e^d\,={\rm Tr}\,(d_\omega(\iota_\xi\omega)\wedge e^2),
\end{align*} 
where we used the second equation of motion in the second equality. 

Using the Cartan identity we have 
    \[
    L_\xi\omega=\iota_\xi d\omega+d\iota_\xi\omega=\iota_\xi(F_{\omega}-\omega\wedge \omega)+d\iota_\xi\omega =\iota_\xi F_\omega+d_\omega(\iota_\xi\omega).
     \] 
     On the other hand, using the first equation of motion, $e\wedge F_\omega=0$, we have $e\wedge\iota_\xi F_\omega=\iota_\xi e\wedge F_{\omega}$, thus ${\rm Tr}(e^2\wedge\iota_{\xi} F_{\omega})=-{\rm Tr}((\iota_\xi e)e\wedge F_\omega)=0$. 
Then, we have 
\[
d\mathcal{Q}[\xi]= {\rm Tr}\left[\frac{1}{2} L_\xi \omega\wedge e^2 - \frac{1}{2} \iota_\xi F_\omega\wedge e^2\right]={\rm Tr}\left[\frac{1}{2} e^2\wedge L_\xi \omega \right],
\]

proving
\[
J_3[\xi]-d\mathcal{Q}[\xi]= {\rm Tr}\left[-\frac{1}{6} e^2\wedge L_\xi \omega - \frac{1}{3} \omega\wedge e\wedge L_\xi e\right].
\]

\end{proof}

\begin{proof}[Proof of Lemma \ref{lem:74}]

We start with any local coframe $e^a$ about a point $p$ with $\xi(p)\neq 0.$ Let $L_\xi e^a={\lambda^a}_b e^b.$ Then, since $\xi$ is a Killing vector, we have 
\[
0=L_\xi g=L_\xi(\eta_{ab} e^a\otimes e^b)=\eta_{ab}(L_\xi e^a\otimes e^b+e^a\otimes L_\xi e^b)=(\lambda_{ab}+\lambda_{ba})\,e^a\otimes e^b,
\] hence  $\lambda_{ab}+\lambda_{ba}=0,$ or $\lambda^T\eta+\eta\lambda=0.$ We can solve the ODE 
\[
\xi({\Lambda^a}_b)+{\Lambda^a}_c{\lambda^c}_b=0, \quad{\rm or}\quad \xi\Lambda=-\Lambda\lambda
\] to get a solution ${\Lambda}$ on an open ball $U\subset M$ with $p\in U.$
We claim we can arrange so that $\Lambda$ has values in $ {\rm SO}(1,3).$  

Indeed, let $K$ be the intersection of the ball $U$ and the affine hypersurface $H$ of codimension one containing $p$ and orthogonal to $\xi(p)$ in the Euclidean inner product. Shrinking the ball if necessary we can assume that $H$ is not parallel to $\xi(q)$ at any $q\in U$. As a Cauchy problem we can impose the initial condition that $\Lambda=I$ on $K.$

For each $q\in K$ we can consider  an integral curve $\gamma(s)$ of $\xi$ with $\gamma(0)=q$. Putting $Q(s)=\Lambda^T(\gamma(s))\,\eta\,\Lambda(\gamma(s))$, we have  $dQ/ds=-\lambda^TQ-Q\lambda.$ Also, $Q(0)=\eta$ and $dQ/ds(0)=0$.  Since the constant $\eta$ is a solution, the uniqueness theorem for ODEs implies $Q(s)=\eta$ for all $s$, which proves that $\Lambda\in {\rm SO}(1,3)$ along the integral curve.  Shrinking if necessary, we may assume that $U$ is foliated by integral curves. This proves our claim.

Now, we define a new coframe by $e'^a={\Lambda^a}_be^b$ and check \[
L_\xi e'^a=(L_\xi{\Lambda^a}_b)e^b+{\Lambda^a}_bL_\xi e^b=[\xi({\Lambda^a}_b)+{\Lambda^a}_c{\lambda^c}_b ]e^b=0.
\] The torsion-free condition, $de'^a+{\omega'^a}_b\wedge e'^b=0$, determines a spin connection $\omega'^{ab}.$
Then, we have 
\[
0=L_\xi(de'^a+{\omega'^a}_b\wedge e'^b )=d(L_\xi e'^a)+(L_\xi {\omega'^a}_b)\wedge e'^b+{\omega'^a}_b\wedge L_\xi e'^b=(L_\xi {\omega'^a}_b)\wedge e'^b. 
\] Let $A_{ab}=L_\xi\omega'_{ab}.$ From $\omega'_{ab}=-\omega'_{ba}$ we know $A_{ab}=-A_{ba}$. Writing $A_{ab}=A_{abc}e'^c,$  we have $A_{abc}=-A_{bac}$. On the other hand, from \[
0=(L_\xi {\omega'^a}_b)\wedge e'^b=A_{ab}\wedge e'^b=A_{abc}e'^c\wedge e'^b
\] we see that $A_{abc}=A_{acb}.$ Hence, we have \[
A_{abc}=-A_{bac}=-A_{bca}=A_{cba}=A_{cab}=-A_{acb}=-A_{abc},
\]which shows that $A_{abc}=0$ for all $a,b,c$. We have proved that  $A_{ab}=L_\xi \omega'_{ab}=0.$

\end{proof}


\subsection{Computations for Schwarzschild Spacetime}

Recall that Schwarzschild spacetime is a static, rotationally symmetric, four-dimensional, vacuum solution of Einstein's equations with a black hole region. More precisely, it is a one-parameter family of such solutions, depending on the Schwarzschild radius $r_S >0$. Fix such a radius $r_S >0$.  We model our four-dimensional Schwarzschild spacetime on the set $M=\RR \times (\RR^3 - \{(0,0,0)\})$ with appropriate metric.  The boundary of the black hole region, the horizon, occurs at $r = r_S$ (in spherical spatial coordinates), and is a \emph{Killing horizon} with Killing field $\xi = \frac{1}{c} \partial_t$. The two-sphere
\[
\Sigma_{r_S} := \{t=0, r=r_S\} \subset \RR \times (\RR^3 - \{(0,0,0)\})
\]
is the spacelike codimension 2 bifurcation surface.

The Schwarzschild radius and the mass ${\bf m}$ are proportional via
\[
r_S = \frac{2G {\bf m}}{c^2},
\]
where $G$ is the gravitational constant and $c$ is the speed of light.  In what follows we will also introduce two further constants: the Boltzmann constant $k_B$ and the reduced Planck constant $\hbar$. Keeping track of these constants doesn't change the mathematics, but does allow us to compare to other computations of various physical quantities. 

We will encode Schwarzschild space via the \emph{static diagonal tetrad} and spin connection.  Let $f(r)$ be the positive square root of the equation
\[
(f(r))^2= 1-\frac{2G {\bf m}}{c^2 r}.
\]
Then, define a tetrad by
\[
e^0 = cf(r) dt, \quad e^1 = \frac{1}{f(r)} dr, \quad e^2 = r d \theta , \quad e^3 = r \sin \theta d \phi,
\]
with associated spin connection $\omega$ given by
\[
\omega^{01}= f' e^0, \quad \omega^{12}= - \frac{f}{r} e^2, \quad \omega^{13}=-\frac{f}{r} e^3, \quad \omega^{23} = - \frac{\cot \theta}{r} e^3,
\]
and subject to $\omega^{ab} = -\omega^{ba}$, with all other components zero. Observe that $e^a$ and $\omega^{ab}$ are regular for $r>r_S$. Using the definition of surface gravity $\xi^{\nu}\Delta_\nu\xi^\mu=\kappa\xi^\mu$ and moving to the Eddington--Finkelstein coordinates $(v,r,\theta,\phi),$\ where $v=t+r^*, dr^*/dr=1/f^2(r),$ we get 
\[
\kappa=\lim_{r\downarrow r_S} f(r)f'(r)=\frac{1}{2r_S}.
\]


Following Wald \cite{Wald}, for an asymptotically flat (stationary) black hole spacetime, we define the \emph{black hole entropy} to be given by
\[
\cS_{BH} := 
\frac{k_B}{\hbar}\frac{2 \pi}{\kappa }\frac{c^3}{16\pi G} {Q}[\xi]_H,
\] where the first factor is for the conversion from the action to the entropy, the second is like $1/T$ in the first law of thermodynamics $dS=dE/T$ and the final factor restores the normalization of the physical Einstein–Hilbert action.  Here ${Q}[\xi]_H$ means the integral of $\mathcal{Q}[\xi]$ over the bifurcation surface $\Sigma_{r_S}.$ Precisely speaking, it is the limit of the integrals of $\mathcal{Q}[\xi]$ over $\Sigma_{r}$ as $r\downarrow r_S.$  
For this class of spacetimes, the Wald entropy agrees with most other definitions of black hole entropy \cite{Wald1}.

We now check the conditions in Lemma 7.4, that is, $L_\xi e=0$ and $L_\xi \omega=0.$ Put $x^\mu=(ct,r,\theta,\phi)$. Then $\xi^\mu=(1,0,0,0)$ and for a 1-form $\alpha=\alpha_\mu dx^\mu$ we have $L_\xi \alpha=(\xi^\nu\partial_\nu \alpha_\mu+\alpha_\nu\partial_\mu \xi^\nu)dx^\mu=(\partial_0\alpha_\mu)dx^\mu=0$ if $\alpha_\mu$ is independent of $t$: observe that $e^a_0$ and $\omega^{ab}_0$ are independent of $t.$ This works in the region $r>r_S,$ where $\xi\neq 0$ and $e^a, \omega^{ab}$ are regular.

Continuing our computation for Schwarzschild (in the static diagonal tetrad), we have
\[
\mathcal{Q}[\xi] =  \frac{1}{2}  \Tr \left (\iota_\xi \omega \wedge e \wedge e \right )= 2\,f(r)f'(r)\, e^2 \wedge e^3 = 2\,r^2f(r)f'(r) \sin \theta \,d\theta \wedge d \phi.
\]
Hence,
\[
\int_{\Sigma_{r}} \mathcal{Q} [\xi] = 2\int_0^{2\pi} \int_0^\pi  r^2 f(r)f'(r)  \sin \theta\, d\theta \wedge d \phi 
= 2f(r)f'(r)\,4\pi r^2
\]and 
\[
{Q} [\xi]_H=\lim_{r\downarrow r_S}\int_{\Sigma_{r}} \mathcal{Q} [\xi]=2\kappa\,{\rm Area}(\Sigma_{r_S}).
\]
Therefore, for the Schwarzschild black hole we have
\[
\cS_{BH} = 
\frac{k_B}{\hbar}\frac{2 \pi}{\kappa }\frac{c^3}{16\pi G} \,2\kappa\,\text{Area}(\Sigma_{r_S})=\frac{k_B c^3}{4\hbar G}  \text{Area}(\Sigma_{r_S})=\frac{k_B }{4}  \frac{\text{Area}(\Sigma_{r_S})}{\ell^2_P},
\]
in accordance with Bekenstein's area law. (Here, $\ell_P$ is the Planck length.)

\subsection{Computations for Kerr Spacetime}

For our Kerr spacetime we use two parameters, ${\bf m}$ for the mass and $J$ for the angular momentum. We put\[
a=\frac{J}{{\bf m}c}, \quad r_S=\frac{2G{\bf m}}{c^2}, \quad \rho^2={r^2+a^2 \cos^2\theta,\,\,\rho>0,} \quad \Delta=r^2+a^2-r_Sr.
\] Then, using the  Boyer--Lindquist coordinates we define the metric
\begin{align*}
ds^2&=-\left(1-\frac{r_Sr}{\rho^2}\right)c^2dt^2-\frac{2r_Sra\sin^2\theta}{\rho^2}cdtd\phi
+\frac{\rho^2}{\Delta}dr^2+\rho^2d\theta^2\\&\quad\quad+\left( r^2+a^2+\frac{r_Sra^2\sin^2\theta}{\rho^2}\right)\sin^2\theta
d\phi^2\\
&=-\frac{\Delta}{\rho^2}\left(cdt-a\sin^2\theta d\phi\right)^2
+\frac{\rho^2}{\Delta}dr^2+\rho^2d\theta^2+\frac{\sin^2\theta}{\rho^2}\left[(r^2+a^2)d\phi-acdt   \right]^2.
\end{align*}
We will use the Carter orthonormal tetrad
\[
e^0=\frac{\sqrt{\Delta}}{\rho}(cdt-a\sin^2\theta d\phi),\,\,\,\,e^1=\frac{\rho}{\sqrt{\Delta}}dr,\,\,\,\,e^2=\rho d\theta,\,\,\,\,e^3=\frac{\sin\theta}{\rho}\left[(r^2+a^2)d\phi-acdt\right].
\]

Upon imposing the torsion-free condition, $de^a+{\omega^a}_b\wedge e^b=0$, we get the spin connection whose nonzero terms are 
\begin{align*}
    \rho^3\omega^{01}&=-\frac{a^2(2r\sin^2\theta+r_S\cos^2\theta)-r_Sr^2}{2\sqrt{\Delta}}\,e^0-ar\sin\theta \,e^3\\
    \rho^3\omega^{02}&=-a^2\sin\theta\cos\theta \,e^0-a\sqrt{\Delta}\cos\theta \,e^3  \\
    \rho^3\omega^{03}&= -ar\sin\theta \,e^1+a\sqrt{\Delta}\cos\theta \,e^2   \\
    \rho^3\omega^{12}&= -a^2\sin\theta\cos\theta \,e^1-r\sqrt{\Delta} \,e^2   \\
    \rho^3\omega^{13}&=  -ar\sin\theta \,e^0-r\sqrt{\Delta} \,e^3  \\
    \rho^3\omega^{23}&= -a\sqrt{\Delta}\cos\theta \,e^0-(r^2+a^2)\cot \theta \,e^3  
\end{align*}  along with the antisymmetric counterparts. Observe that $e^a, \omega^{ab}$ are regular for $r>r_+.$

\vspace{1em}
Observe that when $\Delta=r^2+a^2-r_Sr=0$ we have two horizons. To avoid the extremal case we assume $r_S>2|a|,$ that is, the slow-rotation regime. We put $\rho_+=\rho(r_+,\theta).$ We are interested in the outer horizon $H$ at $r=r_+.$ Let 
\[
\xi=\frac{1}{c}{\partial_ t}+\Omega_H{\partial_\phi} \quad {\rm with} \quad \Omega_H=\frac{a}{r_Sr_+}.
\]  Then $\xi$ is a Killing vector since $g_{\mu\nu}=g_{\mu\nu}(r,\theta).$ Similarly both $e^a$ and $\omega^{ab}$ are of the form $\alpha=\alpha_\mu dx^{\mu},$ where $x^{\mu}=(ct,r,\theta,\phi)$ and  $\alpha_\mu=\alpha_\mu(r,\theta).$ Hence \[(L_\xi\alpha)_\mu=\xi^\nu\partial_\nu\alpha_\mu=\partial_0\alpha_\mu+\Omega_H\partial_3\alpha_\mu=0,\] showing that $L_\xi e^a=0$ and $L_\xi\omega^{ab}=0$ as in Lemma 7.4.

We claim that $H$ is a Killing horizon. First, the normal vector to $H$ is $n_\mu=\partial_\mu r=(0,1,0,0)$ and $n_\mu n^\mu=g^{\mu\nu}\partial_\mu r\partial_\nu r=g^{11}=\Delta/\rho^2=0$ on $H$, hence $H$ is a null hypersurface. And $\xi$ is null on $H$ since 
\begin{align*}
  \xi_\mu\xi^\mu&=g_{00}+2\Omega_H g_{03}+\Omega^2_Hg_{33}\\
  &=-\left(1-\frac{r_Sr_+}{\rho^2_+}\right)-\frac{2\Omega_H r_Sr_+a\sin^2\theta}{\rho^2_+}
+\Omega^2_H\left( r^2_++a^2+\frac{r_Sr_+a^2\sin^2\theta}{\rho^2_+}\right)\sin^2\theta\\&=\rho^{-2}_+\left[
-\rho^2_+ +r_Sr_+-2a^2\sin^2\theta +(r_Sr_+ \rho^2_+ +r_Sr_+a^2\sin^2\theta)\frac{a^2\sin^2\theta}{r^2_Sr^2_+}\right]\\
&=\rho^{-2}_+(a^2\sin^2\theta-2a^2\sin^2\theta+a^2\sin^2\theta)=0.
\end{align*}
Clearly $H$ is  generated by  $\xi$ since $H$ is foliated by curves $x^\mu=\hat{x}^\mu+\lambda\xi^\mu$ with $\hat{x}^1=r_+.$

In order to compute the surface gravity we will use the formula 
\[
-\frac{1}{2}\left(\partial_\mu \xi^2\right)\big|_H=\kappa \xi_\mu\big|_H,
\]which follows from the Killing-horizon identity $\nabla_\mu(\xi^2)=-2\kappa\xi_\mu$ on $H$. However at $r=r_+$ we have a coordinate singularity. So, we introduce the ingoing Kerr coordinates $(v,r,\theta,\psi),$ where, with $\tau=ct,$ we have
\[
dv=d\tau+\frac{r^2+a^2}{\Delta}dr, \quad  d\psi=d\phi+\frac{a}{\Delta}dr.
\]This yields
\[
\begin{aligned}
ds^2={}&
-\left(1-\frac{r_S r}{\rho^2}\right)dv^2
+2\,dv\,dr
+\rho^2 d\theta^2
-\frac{2r_S r a\sin^2\theta}{\rho^2}\,dv\,d\psi  \\
&\quad
-2a\sin^2\theta\,dr\,d\psi
+
\frac{(r^2+a^2)^2-a^2\Delta\sin^2\theta}{\rho^2}
\sin^2\theta\,d\psi^2 .
\end{aligned}
\]

The horizon generator is
\[
\xi=\partial_v+\Omega_H\partial_\psi,
\]
in these new coordinates and 
\[
\begin{aligned}
\xi^2
&=g_{\mu\nu}\xi^\mu\xi^\nu
 = g_{vv}+2\Omega_H g_{v\psi}+\Omega_H^2 g_{\psi\psi} \\
&= -\left(1-\frac{r_S r}{\rho^2}\right)
 +2\Omega_H\left(-\frac{r_S r a\sin^2\theta}{\rho^2}\right)
 +\Omega_H^2\frac{\sin^2\theta}{\rho^2}
 \left[(r^2+a^2)^2-\Delta a^2\sin^2\theta\right].
\end{aligned}
\]
Equivalently,
\[
\begin{aligned}
{\rho^2} \xi^2
&= -{\rho^2}+r_S r-2\Omega_H r_S r a\sin^2\theta
 +\Omega_H^2(r^2+a^2)^2\sin^2\theta
 -\Omega_H^2\Delta a^2\sin^4\theta \\
&= -\Delta+a^2\sin^2\theta
 -\frac{2r_S r}{r_+^2+a^2}a^2\sin^2\theta
 +\left(\frac{r^2+a^2}{r_+^2+a^2}\right)^2a^2\sin^2\theta
 -\Delta\left(\frac{a^2}{r_+^2+a^2}\right)^2\sin^4\theta \\
&= -\Delta\left[1+\left(\frac{a^2}{r_+^2+a^2}\right)^2\sin^4\theta\right]
 +B(r)\,a^2\sin^2\theta\,,
\end{aligned}
\]
where
\[
    B(r)=1-\frac{2r_S r}{r_+^2+a^2}
        +\left(\frac{r^2+a^2}{r_+^2+a^2}\right)^2 .
\]
Observe that
\[
    B(r_+)=0,
    \qquad
    B'(r_+)= -\frac{2r_S}{r_S r_+}
             +\frac{4r_+}{r_S r_+}
           =\frac{2(r_+-r_-)}{r_S r_+}.
\]
Also,
\[
    \Delta=(r-r_+)(r-r_-)
          =(r-r_+)(r_+-r_-)+O(\Delta^2).
\]
Hence
\[
    B(r)=B(r_+)+B'(r_+)(r-r_+)+O((r-r_+)^2)
         =\frac{2\Delta}{r_S r_+}+O(\Delta^2).
\]
Now,
\[
    {\rho^2} \xi^2
    =-\Delta\left(1-\frac{a^2\sin^2\theta}{r_+^2+a^2}\right)^2
      +O(\Delta^2)
    =-\Delta\left(\frac{\rho^2_+}{r_S r_+}\right)^2
      +O(\Delta^2).
\]
Near $H$ we have
\[
    \frac{1}{\rho^2}
    =\frac{1}{\rho^2_+ + O(r-r_+)}
    =\frac{1}{\rho^2_+ + O(\Delta)}
    =\frac{1}{\rho^2_+}+O(\Delta),
\]
thus
\[
    -\xi^2
    =\Delta\frac{{\rho^2}_+}{(r_S r_+)^2}+O(\Delta^2)
\]
and
\[
    \partial_r\left(-\xi^2\right)\big|_{H}
    =(r_+-r_-)\left(\frac{{\rho}_+}{r_S r_+}\right)^2.\]
   Finally, using 
\[
\xi_r=g_{rv}\xi^v+g_{r\psi}\xi^\psi=1-a\sin^2\theta\,\cdot\Omega_H=1-a\sin^2\theta\,\cdot\frac{a}{r_Sr_+}=\frac{\rho^2_+}{r_Sr_+},
\] we get   
   \[\kappa
    =\frac{
-\frac{1}{2}\left(\partial_\mu \xi^2\right)\big|_H}{ \xi_\mu\big|_H}
    =\frac{\frac{1}{2}(r_+-r_-)\left(\frac{{\rho}_+}{r_S r_+}\right)^2}{\frac{\rho^2_+}{r_Sr_+}}=\frac{r_+-r_-}{2r_S r_+}.
\]

The ingoing Kerr coordinates used above are regular on the future event horizon $r=r_+$, and they allow us to verify that $\xi$ is null on $r=r_+$ and to compute the surface gravity. However, these coordinates do not display the bifurcation surface. To see that the horizon generator vanishes on the bifurcation surface, one must pass to Kruskal-type coordinates $(U,V,\theta,\widetilde{\psi})$ adapted to the horizon. In such coordinates the future and past horizons are given by $U=0$ and $V=0$, respectively, while the bifurcation surface is $\Sigma=\{U=V=0\}$. The Killing generator takes the local boost form
\[
\xi = \kappa\left(V\frac{\partial}{\partial V} - U\frac{\partial}{\partial U}\right),
\]
and therefore $\xi|_{\Sigma}=0$. We will not go into the details of Kruskal-type coordinates construction here.

In order to compute the surface charge we return to the Boyer--Lindquist coordinates. From our tetrad
we get
\[
    \iota_\xi e^0
    =\frac{\sqrt{\Delta}}{\rho}
      \left(1-a\Omega_H\sin^2\theta\right),
    \,\,\,
    \iota_\xi e^1=0,
    \,\,\,
    \iota_\xi e^2=0,
    \,\,\,
    \iota_\xi e^3
    =\frac{\sin\theta}{\rho}
      \left[(r^2+a^2)\Omega_H-a\right].
\]
Here $\iota_\xi e^0,\iota_\xi e^3\to 0$ as $r\to r_+$. Observe that
$\omega^{ab}=P^{ab}_{c}e^c,
$
where $P^{ab}_{c}=P^{ab}_{c}(r,\theta)$. Hence, $
    \iota_\xi\omega^{ab}
    =P^{ab}_{0}\,\iota_\xi e^0+P^{ab}_{3}\,\iota_\xi e^3.$
All $P^{ab}_c$ have finite limits as $r\to r_+$ except $P^{01}_0$.
The only nonzero limiting contribution to $\iota_\xi\omega^{ab}$ comes from 
\[
\begin{aligned}
\iota_\xi\omega^{01}
&=
-\frac{a^2(2r\sin^2\theta+r_S\cos^2\theta)-r^2r_S}
       {2\rho^3\sqrt{\Delta}}\,
 \frac{\sqrt{\Delta}}{\rho}
 \left(1-\frac{a^2}{r_S r_+}\sin^2\theta\right) \\
&=
\frac{-2a^2r +(2r-r_S)a^2\cos^2\theta+r^2r_S}
     {2\rho^4}\,
 \frac{\rho^2_+}{r_S r_+} \\
&=
\frac{(2r-r_S)(a^2\cos^2\theta+r^2)}
     {2\rho^{-2}_+\rho^4r_S r_+}
 \longrightarrow\frac{r_+-r_-}{2r_S r_+}
 =\kappa 
\end{aligned}
\] as $r\downarrow r_+.$
Hence, with $\Sigma_r=\{t=0,r={\rm constant}\}$ we have 
\begin{align*}
Q[\xi]_H
&=\lim_{r\downarrow r_+}\int_{\Sigma_r}\frac12\epsilon_{abcd}\,
  \iota_\xi\omega^{ab}\,e^c\wedge e^d \\
&=2\int_{\Sigma_{r_+}}\iota_\xi\omega^{01}\,e^2\wedge e^3 \\
&=2\int_0^\pi\int_0^{2\pi}
  \kappa\,\rho_+\,d\theta\,\,
  \frac{\sin\theta}{\rho_+}(r_+^2+a^2)\,d\phi \\
&=2\kappa\, {\rm Area}(\Sigma_{r_+}).
\end{align*}
The Wald entropy is then
\[
    \cS_{BH}
    =\frac{c^3}{16\pi G}\,\frac{2\pi}{\kappa}\frac{k_B}{\hbar}\,Q[\xi]_H
    =
    \frac{k_B}{4}\frac{{\rm Area}(\Sigma_{r_+})}{\ell_P^2}
\]
and is again in accordance with Bekenstein's area law.

\bibliographystyle{plain}
\bibliography{adGR_bib}

\end{document}